\documentclass{article}
\usepackage{fullpage}
\usepackage{epsfig}
\usepackage{graphics}
\usepackage{latexsym}
\usepackage{amsmath}
\usepackage{amsfonts}
\usepackage{amssymb}
\usepackage[mathscr]{eucal}
\usepackage{mathrsfs}
\usepackage{pifont}
\usepackage{yhmath}
\usepackage{undertilde}
\usepackage[usenames]{color}
\usepackage{amsthm}

\usepackage{url}
\urldef{\mailsa}\path|{xchen,xdhu,wcj}@amss.ac.cn|

\newtheorem{theorem}{Theorem}[section]
\newtheorem{lemma}[theorem]{Lemma}

\newtheorem{corollary}[theorem]{Corollary}

\newtheorem{algorithm}{Algorithm}{\itshape}{\rmfamily}

\newtheorem{observation}{Observation}

 \theoremstyle{remark}

{\itshape}{\rmfamily}

{\itshape}{\rmfamily}

\newtheorem{Observation}{Observation}{\itshape}{\rmfamily}

{\itshape}{\rmfamily}

\makeatletter
\@addtoreset{equation}{section}
\def\section{\@startsection {section}{1}{\z@}{-3.5ex plus -1ex minus
 -.2ex}{2.3ex plus .2ex}{\large\bf}}
\makeatother

\def\bfm#1{\mbox{\boldmath$#1$}}

\def\0{\bfm 0}

\newcommand{\bluecomment}[1]{\textcolor{black}{\textrm{#1}}}

\DeclareMathAlphabet{\mathpzc}{OT1}{pzc}{m}{it}

\def\bfm#1{\mbox{\boldmath$#1$}}

\begin{document}


\title{\bf Approximability of the Minimum Weighted Doubly Resolving Set Problem \thanks{Research supported in part by NNSF of China under Grant No. 11222109, 11021161 and 10928102,
by 973 Project of China under Grant No. 2011CB80800,   and by CAS Program for Cross \& Cooperative Team of Science  \& Technology Innovation.}}


 \author{Xujin Chen \and Xiaodong Hu \and  Changjun Wang}

\date{Academy of Mathematics and Systems Science \\ Chinese Academy
of Sciences, Beijing 100190, China\\
${}$\\
\mailsa}

\maketitle

\begin{abstract} {Locating source of diffusion in networks is crucial for controlling and preventing epidemic   risks. It has been studied under various probabilistic models. In this paper, we study source location  from a deterministic point of view by modeling it as the minimum weighted doubly resolving set (DRS) problem, which is a strengthening of the well-known metric dimension problem.}

 \hspace{4mm}Let $G$ be a vertex weighted undirected graph on $n$ vertices. A vertex subset $S$ of $G$ is  DRS  of $G$ if for every pair of vertices $u,v$ in $G$, there exist $x,y\in S$ such that the difference of distances (in terms of number of edges) between $u$ and $x,y$ is {not equal to} the difference of distances  between $v$ and $x,y$. The minimum weighted DRS problem consists of finding a DRS in $G$  with  minimum total weight. 
 We establish $\Theta(\ln n)$ approximability of the minimum DRS problem on general graphs for both weighted and unweighted versions. This is the first work providing explicit approximation lower and upper bounds for minimum (weighted) DRS problem, which are nearly tight. Moreover, we design first known strongly polynomial time algorithms for the minimum weighted DRS problem on   general wheels and trees with additional constant $k\ge0$ edges.
\end{abstract}

\noindent{\bf Keywords:} {Source location, Doubly resolving set, Approximation  algorithms, Polynomial-time solvability, Metric dimension

\newcounter{my}
\newenvironment{mylabel}
{
\begin{list}{(\roman{my})}{
\setlength{\parsep}{-0mm}
\setlength{\labelwidth}{8mm}
\setlength{\leftmargin}{8mm}
\usecounter{my}}
}{\end{list}}

\newcounter{my2}
\newenvironment{mylabel2}
{
\begin{list}{(\alph{my2})}{
\setlength{\parsep}{-1mm} \setlength{\labelwidth}{12mm}
\setlength{\leftmargin}{14mm}
\usecounter{my2}}
}{\end{list}}

\newcounter{my3}
\newenvironment{mylabel3}
{
\begin{list}{(\alph{my3})}{
\setlength{\parsep}{-1mm}
\setlength{\labelwidth}{8mm}
\setlength{\leftmargin}{10mm}
\usecounter{my3}}
}{\end{list}}

\section{Introduction}

{Locating the source of a diffusion in  complex networks is an intriguing challenge, and finds diverse applications in controlling and preventing network epidemic risks \cite{sz11}. In particular, it is often financially and technically
impossible to observe the state of all vertices in a large-scale network, and, on the other hand, it 
  is desirable to find the location of the source (who initiates the diffusion) from measurements collected by sparsely placed {observers} \cite{ptv12}. Placing an observer at vertex $v$ incurs a cost, and the observer with a clock   can record the time at which the state of $v$ is changed (e.g., 
knowing a rumor, being infected or contaminated). 
 Typically, the time when the {single} source originates an information is unknown \cite{ptv12}. The observers  can only report  the times they receive the information, but the senders of the information (i.e., we do not know who infects whom or who influences whom) \cite{glk10}. The information is diffused from the source to any vertex through shortest paths in the network, i.e., as soon as a vertex receives the information, it sends the information to all its neighbors simultaneously, which takes one time unit. 
Our goal is to select a subset $S$ of vertices with minimum total cost {such that} 
the source can be uniquely located by the ``infected" times of vertices in $S$. This 
problem is equivalent to finding a minimum weighted {\em doubly resolving set} (DRS) in networks defined as follows.}  

 \paragraph{DRS model.}  {Networks are modeled as undirected connected graphs without parallel edges nor loops}.  Let $G=(V,E)$ be a     graph on $n\ge2$ vertices, and each vertex $v\in V$ has a nonnegative {\em weight} $w(v)$, {representing its cost}. For any $S\subseteq V$, the {\em weight} of $S$ is defined to be $w(S):=\sum_{v\in S}w(v)$. For any $u,v\in V$, we use $d_G(u,v)$ to denote the {\em distance} between $u$ and $v$ {in $G$}, i.e., the number of edges in a  shortest path between $u$ and $v$. Let $u,v,x,y$ be four distinct vertices of $G$. Following C\'aceres et al. \cite{cjmppsw07}, we say that $\{u,v\}$ {\em doubly resolves} $\{x,y\}$, or $\{u,v\}$ {\em doubly resolves} $x$ and $y$, if \[d_G(u,x)-d_G(u,y)\ne d_G(v,x)-d_G(v,y).\] Clearly $\{u,v\}$ doubly resolves $\{x,y\}$ if and only if $\{x,y\}$ doubly resolves $\{u,v\}$. {For any subsets $S,T $ of vertices,  $S$} {\em doubly resolves} $T$ if every pair of vertices in $T$ is doubly resolved by some pair of vertices in $S$. {In particular, $S$ is called} a {\em doubly resolving set} (DRS) of $G$ if $S$ doubly \mbox{resolves $V$.} Trivially, $V$ is a \mbox{DRS of $G$.} The {\em minimum weighted doubly resolving set} (MWDRS) problem is to find a DRS of $G$ that has a minimum weight {(i.e. a minimum weighted DRS of $G$)}. In the special case where all vertex weights are equal to $1$, the problem is referred to as the {\em minimum doubly resolving set} (MDRS) problem {\cite{mkkm12}},  and it concerns with the minimum cardinality ${\tt dr}(G)$ of    DRS of $G$.

 {Consider arbitrary $S\subseteq V$. It is easy to see that   $S$ fails to locate the diffusion source in $G$} at some case if and only if there exist distinct vertices $u,v\in V$ such that $S$ cannot distinguish between the case of $u$ being the source and that of $v$ being the source, i.e., $d_G(u,x)-d_G(u,y)=d_G(v,x)-d_G(v,y)$ for any $x,y\in S$; equivalently, $S$ is not a DRS of $G$. \bluecomment{(See Appendix \ref{apx:appl}.)} Hence, the MWDRS problem models exactly the problem of finding cost-effective observer {placements} for locating source, as mentioned in our opening paragraph.      

\paragraph{Related work.}
 {Epidemic diffusion and information cascade in networks has been extensively studied for decades in efforts to understand  the diffusion dynamics and its dependence on various factors, such as    network structures and infection rates. However, the inverse problem of inferring the source of diffusion based on limited observations is far less studied, and was first tackled by Shah and Zaman \cite{sz11} for identifying the source of rumor, where the rumor flows on edges according to independent exponentially distributed random times. A maximum likelihood (ML) estimator was proposed for maximizing the correct
localizing probability, and the notion of rumor-centrality was developed for approximately tracing back the source from the configuration of infected vertices at a given moment. The accuracy of estimations heavily depended on the structural properties of the networks. Shah and Zaman's model and their results on trees were   extended by Karamchandani and Franceschetti \cite{kf03}  to the case in which nodes reveal whether they have heard the rumor with independent probabilities. Along a different line, Pinto et al. \cite{ptv12} proposed  other ML estimators that perform  source detection via sparsely distributed observers who measure from which neighbors and at what time they received the information. The ML estimators were shown to be optimal for trees, and suboptimal for general networks under  the assumption that the propagation delays associated with edges are i.i.d. random variables with known Gaussian distribution. 
 In contrast to previous probabilistic model for estimating the location of the source,    we study the problem from a combinatorial optimization's point of view; our goal is to find   an  observer set of minimum cost that guarantees   deterministic  determination of the accurate location of the source, i.e., to find a minimum weighted DRS.}

 The double resolvability is a strengthening of the  well-studied resolvability, where a vertex {$x$} resolves two vertices {$u,v$} if and only if {$d_G(u,x)\ne d_G(v,x)$.} A subset $S$ of $V$ is a {\em resolving set} (RS) of $G$ if every pair of vertices is resolved by some vertex of $S$. The minimum cardinality of a RS of $G$ is known as the {\em metric dimension} ${\tt md}(G)$ of $G$, which has been extensively  studied due to its theoretical importance and diverse applications (see e.g., \cite{cjmppsw07,
cz03,elw12,hsv12} and references therein). Most literature on finding minimum resolving sets, {known as the {\em metric dimension problem},} considered the unweighted case. 
The {unweighted} problem is $NP$-hard even for planar graphs,  split graphs, bipartite graphs and bounded degree graphs \cite{dpsv12,elw12,hsv12}. On general graphs, Hauptmann et al. \cite{hsv12} showed that the unweighted problem is not approximable within $(1-\varepsilon)\ln n$ for any $\varepsilon > 0$, unless $NP\subset
 DTIME(n^{\log \log n})$; moreover, the authors \cite{hsv12} gave a $(1+o(1))\ln n $-approximation algorithm based on approximability results of the test set  problem in bioinformatics \cite{bdk05}. A lot of research efforts have been devoted to {obtaining} the exact values or upper bounds of the metric {dimensions} of special graphic classes~
 {\cite{cjmppsw07}}. Recently, Epstein et al \cite{elw12} studied the weighted version of the problem, and developed polynomial time exact algorithms for finding a minimum weighted  RS, when the underlying graph $G$ is  a cograph, a {\em $k$-edge-augmented tree} (a tree with additional $k$ edges) for constant $k\ge0$, or a (un)complete wheel.

Compared with nearly four decade research and vast {literatures} on   resolving sets (metric dimension), the study on DRS has a {relatively} short history and its results   have been very limited. The concept of DRS was introduced in 2007  by C\'aceres et al. \cite{cjmppsw07}, who proved that the minimum {RS} of the Cartesian product of graphs is tied in a strong sense to minimum DRS of the graphs:  the metric dimension of the Cartesian product of graphs $G_1$ and $G_2$ is upper bounded by   ${\tt md}(G_1)+{\tt dr}(G_2)$.
When restricted to the same graph, it is easy to see that a DRS must be a {RS}, but the reverse is not necessarily true
. Thus ${\tt md}(G)\le{\tt dr}(G)$. The ratio ${\tt dr}(G)/{\tt md}(G)$ can be arbitrarily large. This can be seen from the tree graph $G$ depicted in Fig. \ref{fg:example}. On the one hand, it is easily checked that  $\{r_1,r_2\}$ is   a {RS} of $G$, giving  ${\tt dr}(G)\le 2$. On the other hand, ${\tt md}(G)=n/2$ since $\{s_1,s_2,\ldots,s_h\}$ is the unique minimum DRS of $G$, as proved later in Lemma \ref{tree} of this paper. In view of the large gap, algorithmic study on DRS deserves good 
efforts, and it is interesting to explore the algorithmic relation between the minimum (weighted) DRS problem and \mbox{its resolving set counterpart.}

\begin{figure}[h]
\centerline{
  \includegraphics[width=6cm]{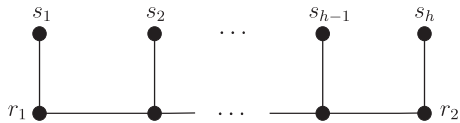}}
  \caption{The graph tree $G$ with ${\tt dr}(G)=n/2$ and ${\tt md}(G)=2$.}\label{fg:example}
\end{figure}

Previous research on DRS considered only the unweighted case. As far as general graphs are concerned, the MDRS problem has been proved to be $NP$-hard \cite{kck09}, and solved experimentally by metaheuristic approaches that {use} binary encoding and standard genetic operators \cite{kck09} and that use variable neighborhood search \cite{mkkm12}.
To date, no efficient general-purpose algorithms with theoretically provable performance guarantees have been developed for the MDRS problem, let alone the MWDRS problem. Despite the NP-hardness, the approximability status of either problem has been unknown in literature. For special graphs, it is known that every RS of {Hamming graph} 
is also a DRS \cite{kkcs12}. Recently, \v{C}angalovi\'{c} et al. showed that 
${\tt dr}(G)\in\{3,4\}$ when $G$ {is a prism graph \cite{ckks13} or} belongs to one of two classes of convex polytopes \cite{kkcs2012}.

 \paragraph{Our contributions.} {As far as we know, our opening example of cost-effective source location is the first  real-world application of  {DRS} explicitly addressed. Motivated by the application, we study and provide a thorough treatment of the MWDRS problem in terms of algorithmic approximability. Broadly speaking, we}   show that the 
 {MDRS and MWDRS  problems have similar approximability to their resolving set counterparts}.

 Based on the construction of Hauptmann et al.~\cite{hsv12}, we prove that there is an approximation preserving reduction from the minimum dominating set problem to the {MDRS} problem, showing that the MDRS problem does not admit $(1-\varepsilon)\ln n$-approximation algorithm for any $\varepsilon>0$ unless $NP\subset DTIME(n^{\log\log n})$. The strong inapproximability improves the $NP$-completeness established in \cite{kck09}. {Besides, we develop a $(\ln n+\ln\log_2n+1)$-approximation algorithm for solving the MWDRS problem in $O(n^4)$ time, based on a modified version of the approximation algorithms used in \cite{bdk05,hsv12}.} 
 To the best of our knowledge, this paper is the first work providing explicit approximation lower and upper bounds for {the MDRS and MWDRS problems}, which are nearly tight {(for large $n$)}.  {A byproduct of our algorithm gives the first logarithmic approximation for the weighted metric dimension problem on general graphs.}

 {Despite many significant technical differences between handling DRS  and RS}, we establish the polynomial time solvability of the MWDRS problem for  all \bluecomment{these} graph classes, with one exception of cographs, where the {weighted metric dimension problem} is known to admit efficient exact algorithms \cite{elw12}. Our results are first known strong polynomial time algorithms for the MDRS problem  on  $k$-edge-augmented trees  and general wheels, including paths, trees and cycles. Using the fact that every minimum weighted DRS  is {\em minimal} (with respect to the inclusion relation), our algorithms make use of the graphic properties to cleverly ``enumerate'' minimal doubly resolving sets that are potentially minimum weighted, and select the best one among them. 

 \medskip The paper is organized as follows: The inapproximability is proved in Section~\ref{sec:inapx}, The approximation algorithm for general graphs and exact algorithms for special graphs are presented in Sections \ref{sec:apx} and   \ref{sec:exact}, respectively. Future research directions are discussed in Section \ref{sec:conclude}. The omitted details are given in Appendix.

\section{Approximation lower bound}\label{sec:inapx}
 {In this section, we establish a logarithmic lower bound for approximation the MDRS problem under the assumption that $NP\not\subset DTIME(n^{\log\log n})$. {Hauptmann et al. \cite{hsv12}} constructed a reduction from the {\em minimum dominating set} (MDS) problem to the metric dimension   problem. {Although their proof does not work  for DRS,}   we show that their construction actually provides an approximation preserving reduction from the MDS problem to the MDRS problem.}

 A vertex subset $S$ of graph $G$ is a {\em dominating set} of $G$ if every vertex outside $S$ has a neighbor in $S$.   {The MDS problem is to find a dominating set of $G$ that has the minimum cardinality ${\tt ds}(G)$. Unless $NP\subset DTIME(n^{\log\log n})$, the MDS problem 
 cannot be approximated within $(1-\varepsilon)\ln n$ for any $\varepsilon>0$ \cite{cc08
 }.}

\begin{lemma}\label{lem:np1}
There exists a polynomial time  transformation  that {transfers graph $G=(V,E)$ 
to graph $G'=(V',E')$ 
such that}   ${\tt dr}(G')\le {\tt ds}(G)+\lceil\log_2 n\rceil+3$.\qed
\end{lemma}

 {Let graphs $G$ and $G'$ be as in Lemma \ref{lem:np1}. It has been shown that, given any  RS (in particular DRS) $S$ of $G'$,  a dominating set of $G$ with cardinality at most $|S|$ can be found in polynomial time \cite{hsv12}.   
{This, in combination with Lemma \ref{lem:np1}  and the logarithmic inapproximability} of the MDS problem \cite{cc08
}, gives the following lower bound for approximating minimum DRS.}
\begin{theorem}\label{th:lower}
Unless $NP\subset DTIME(n^{\log(\log n)})$, the {MDRS} problem cannot be approximated in polynomial time   within a factor of $(1-\epsilon)\ln n$, for any 
 $\epsilon>0$.\qed
\end{theorem}
\section{Approximation algorithm}\label{sec:apx}
In this section, we {present} an $O(n^4)$ time approximation algorithm for the MWDRS problem in general graphs {that achieves} approximation ratio $(1+o(1))\ln n$, {nearly matching the lower bound $\ln n$ established in Theorem \ref{th:lower}.

Our algorithm {uses similar idea to that of Hauptmann et al.} \cite{hsv12} for approximating minimum resolving sets {in the metric dimension (MD) problem. The MD problem is a direct ``projection'' of the {\em unweighted} test set problem studied by Berman et al. \cite{bdk05} in the sense that a vertex in the MD problem can be seen as a ``test" in the test set problem, which allows Hauptmann et al. to apply Berman-DasGupta-Kao algorithm \cite{bdk05} directly. However, in the DRS problem, one cannot   simply} view two vertices as a ``test'', {because such a ``test'' would fail the algorithm in some situation.} {Besides,} the algorithm deals with only unweighted cases. Thus we {need  conduct certain transformation that transforms the DRS problem to a series of weighted test set problems. Furthermore, we need modify Berman-DasGupta-Kao algorithm to solve these {\em weighted} problems within logarithmic approximation ratios.}

\paragraph{Transformation.} For any $x\in V$, let $U_x=\{\{x,v\}:v\in V\setminus \{x\}\}$. {As seen later, each element of $U_x$ can be viewed as a {\em test} or a certain combination of {\em tests} in the test set  problem studied in \cite{bdk05}. From this point of view, we call each element of $U_x$ a {\em super test}, and consider the  {\em minimum weighted  super test set} (MWSTS) problem on $(V,U_x)$} as follows: For {each super} test $T=\{x,v\}\in U_x$, let its weight  be $w(T)=w(v)$, The problem is to find
  a set of super tests $\mathcal T\subseteq U_x$  such that each pair of vertices in $G$ is doubly resolved by some super test in $\mathcal T$ and the weight $w(\mathcal T)=\sum_{T\in\mathcal T}w(T)$ of $\mathcal T$ is minimized. {The following lemma establishes the relation between the MWDRS problem and the MWSTS problem}.

 \begin{lemma}\label{lem:star}
Let $S$ be a {DRS} of $G$ and $s\in S$. Then every pair of vertices in $G$ is {doubly} resolved by at least one element of $\{\{s,v\}:v\in S \setminus\{s\}\}$.
\end{lemma}
\begin{proof} Let $u,v$ be any two distinct vertices of $G$. There {exist} $s_1,s_2\in S $ {such that}   $d_G(u,s_1)-d_G(v,s_1)\ne d_G(u,s_2)-d_G(v,s_2)$. It follows that either $d_G(u,s_1)-d_G(v,s_1)\ne d_G(u,s)-d_G(v,s)$ or $ d_G(u,s)-d_G(v,s)\ne d_G(u,s_2)-d_G(v,s_2)$, saying that $u$ and $v$ are doubly resolved by either $\{s,s_1\}$ or $\{s,s_2\}$.
\end{proof}
Since $V$ is a  {DRS} of $G$, Lemma \ref{lem:star} implies that  {$U_x$ doubly resolves $V$}. More importantly, Lemma \ref{lem:star} provides the following immediate corollary {that is crucial to our algorithm design.}
\begin{corollary}\label{cor:bound}
Let $S^*$ be a {minimum weighted} DRS of $G$ and {$\alpha\in S^*$.  Then the} minimum weight of a solution to the MWSTS problem on $(V,U_{\alpha})$ is at most $w(S^*)-w(\alpha)$.\qed
\end{corollary}


\paragraph{Approximation.} {In order to solve the {MWSTS} problem, we adapt Berman-DasGupta-Kao algorithm \cite{bdk05} to augment a set $\mathcal T $ ($\subseteq U_x$)  of super tests to be a feasible solution step by step.} We define  {\em equivalence relation} $\equiv^{\mathcal T}$ on $V$   by: two vertices {$u,v\in V$} are equivalent under $\equiv^{\mathcal T}$  if and only if {$\{u,v\}$  is not doubly resolved by  any test of  $\mathcal T$.} Clearly, the number of equivalence classes is non-decreasing with the size of $\mathcal T$. Let
  $E_1,\ldots,E_k$ be the equivalence classes of $\equiv^{\mathcal T}$.  The value
    $H_{\mathcal T}:=\log_2(\prod_{i=1}^{k}|E_i|!)$ is called the {\em entropy} of $\mathcal T$. 
    {Note that
    \begin{eqnarray}\label{eqv}
    \begin{array}{rcl}
    H_{\mathcal T}=0&\Leftrightarrow&\text{every equivalent class of $\equiv^{\mathcal T}$ is a singleton}\\
    &\Leftrightarrow&\cup_{T\in \mathcal T}T\text{ is a DRS of $G$.}
    \end{array}
    \end{eqnarray}
     Hence our task is reduced to finding a set $\mathcal T$ of super tests with zero entropy $H_{\mathcal T}$ and weight $w(\mathcal T)$ as small as possible.}  

    {For any super test   $T\in U_x $}, 
   an equivalence class   of $\equiv^{\mathcal T}$ is either an equivalence class of $\equiv^{\mathcal T\cup T}$ or it is partitioned into several {(possibly more than two)} equivalence classes of $\equiv^{\mathcal T\cup T}$. {(If $T$ partitions each equivalent class into at most two equivalent classes, then $T$  works as a test in the test set problem.)} Therefore $H_{\mathcal T}\ge H_{\mathcal T\cup T}$, and $IC(T,\mathcal T):=H_{\mathcal T}-H_{\mathcal T\cup T}\ge0$ 
     equals  the decreasing amount of the entropy when adding $T$ to $\mathcal T$. It is clear that
     \begin{gather}\label{bound}
     IC(T,\emptyset)\le \log_2n!-log_21< n\log_2n.\end{gather}

We  {now give a $(1+o(1))\ln n$-approximation algorithm  for the {MWSTS} problem  on $(V,U_x)$}.  The algorithm adopts the   greedy heuristic  to decrease the entropy of the current set of super tests  \mbox{at a minimum cost (weight).}

\medskip
\hrule
\vspace{-1mm}\begin{algorithm}\label{alg1}
Finding minimum weighted set $\mathcal T$ of super sets.
\vspace{1mm}\hrule
{\small
\begin{enumerate}
  \vspace{-0mm}\item \hspace{1mm}  {$\mathcal T\leftarrow\emptyset$}
  \item \hspace{1mm} \textbf{while} $H_{\mathcal T}\neq 0$ \textbf{do}
  \item \hspace{5mm} Select a {super test} $T\in U_x-\mathcal T$ that \emph{maximizes} $\frac{IC(T,\mathcal T)}{w(T)}$
  \item \hspace{5mm} {$\mathcal T\leftarrow\mathcal T\cup T$}
  \item \hspace{1mm} \textbf{end-while}
\end{enumerate}
}
\vspace{-0mm}\hrule
\end{algorithm}

The major difference between Algorithm \ref{alg1} and the algorithms in \cite{bdk05,hsv12} is the criterion used in Step 3 for selecting $T$. It generalizes the previous   unweighted setting.
The following lemma
extends the result 
on test set \cite{bdk05} to super test set.
\begin{lemma}\label{ich}
  $IC(T,\mathcal T_0)\ge IC(T,\mathcal T_1)$ for any sets $\mathcal T_0$ and $\mathcal T_1$ of super tests with \mbox{$\mathcal T_0\subseteq\mathcal T_1$.}\qed
\end{lemma}

Using (\ref{bound}) and Lemma \ref{ich}, the proof of performance ratio goes almost verbatim as the argument of Berman et al. \cite{bdk05}. We include a proof in Appendix \ref{apx1} for completeness.

\begin{theorem}\label{th:performance}
Algorithm \ref{alg1} is an $O(n^3)$ time   algorithm for {the MWSTS   problem} on $(V,U_x)$ with {approximation} ratio {\mbox{$\displaystyle\ln\left( \max_{T\in U_x}IC(T,\emptyset)\right)\!+\!1\!\le\!\ln n\!+\!\ln\log_2 n\!+\!1$.$\Box$}}
\end{theorem}

Suppose that given the MWSTS   problem on $(V,U_x)$, Algorithm \ref{alg1}  outputs a super test set $\mathcal T_x$. By (\ref{eqv}), 
  Running Algorithm \ref{alg1} for $n$ times, we obtain $n$ {doubly resolving  sets} $S_x$, $x\in V$ of $G$, from which we select the one, say $S_v$, that has the minimum weight,  
i.e. $w(S_v)=\min\{w(S_x):x\in V\}$.
\begin{theorem}
The {MWDRS} problem can be approximated in $O(n^4)$ time within a ratio \[\displaystyle\ln\left( \max_{u,v\in V}IC(\{u,v\},\emptyset)\right)+1\le \ln n+\ln\log_2n+1=(1+o(1))\ln n.\]
\end{theorem}
\begin{proof}
Let $S^*$ be an optimal solution to the MWDRS problem. It suffices to show $w(S_v)/w(S^*)\le (1+o(1))\ln n$. Take $\alpha\in S^*$, and let $\mathcal T^*_{\alpha}$ be an optimal solution to the MWSTS problem on $(V,U_{\alpha})$. It follows from {the} choice of $S_v$, Theorem \ref{th:performance} and Corollary \ref{cor:bound} that $w(S_v)\le w(S_{\alpha})=w(\alpha)+w(\mathcal T_{\alpha})\le w(\alpha)+(\ln n+\ln\log_2n+1)  w(\mathcal T_{\alpha}^*)<(\ln n+\ln\log_2n+1) w(S^*)$.  
\end{proof}

Our algorithm and analysis show that the {algorithm} 
of   \cite{bdk05} can be extended to solve the weighted test set problem, where each test has a nonnegative weight, by changing the selection criterion to be maximizing $IC(T,\mathcal T)$ divided by the weight of $T$. A similar extension applied to the algorithm of Hauptmann et al.~\cite{hsv12} gives a $(1+o(1))\ln n$-approximate solution to the {weighted metric dimension} problem.

\section{Exact algorithms}\label{sec:exact}
{Let $k\ge0$ be a constant. A connected graph is called a {{\em $k$-edge-augmented tree}} if the removal of  {at most $k$ edges} from the graph leaves a spanning tree. Trees and cycles are 0-edge- and {1-edge-augmented} trees, respectively. We design efficient algorithms for solving the MWDRS problem {exactly} on $k$-edge-augmented trees.}
 Our algorithms run in linear time for $k=0,1$, and 
in $O(n^{12k})$ time for $k\ge2$.

{A graph is called a general wheel if it is formed from a cycle by adding a vertex and joining it to some (not necessarily all) vertices on the cycle. We solve the MWDRS problem
on general wheels   in cubic time by dynamic programming.}



\subsection{$k$-edge-augmented trees} Let $G=(V,E)$ be a $k$-edge-augmented tree, and let $L$  be the set of leaves (degree 1 vertices) in $G$. For simplicity, we often use $d(u,v)$ instead of $d_G(u,v)$ to denote the distance between vertices $u,v\in V$ in the underlying graph $G$ of the MWDRS problem.

\subsubsection{Trees: the case of $k=0$.}
When $k=0$, graph $G=(V,E)$ is a tree.  {There is a fundamental difference between DRS and RS of $G$ in terms of minimal sets. In general, $G$ may have multiple minimal RSs and even multiple minimum weighted  RSs. Nevertheless, in any case $G$}   has only one {minimal} DRS, {which consists of all its leaves}.  In particular, we have ${\tt dr}(G)=|L|$.
\begin{lemma}\label{tree}
$L$ is the unique minimal DRS of $G$.
\end{lemma}
\begin{proof}
For any two vertices $u,v\in V$, there exist leaves $l_1,l_2\in L$ such that  the path between $l_1$ and $l_2$ goes through $u$ and $v$. It is easy to see that 
$d(u,l_1)-d(u,l_2)\neq d(v,l_1)-d(v,l_2)$. So $L$ is a DRS. On the other hand, consider any leaf $l\in L$ and its neighbor $p\in V$. Since $d(l,v)-d(p,v)=1$ for any $v\in V-\{l\}$, we see that each DRS of $G$ contains $l$, and thus $L$. The conclusion follows.
\end{proof}

\subsubsection{Cycles: a special case of  $k=1$.}\label{cyc}
Let $G=v_1v_2\cdots v_nv_1$ be a cycle, where  $V=\{v_1,v_2,\ldots,v_n\}$. Suppose without loss of generality that $w(v_p)=\min_{i=1}^nw(v_i)$, {where $p:=\lceil n/2\rceil$.} 
{It was known  that any pair of vertices whose distance is not exactly $n/2$ is a minimal RS of $G$, and vice versa \cite{elw12}.}   {As the next lemma shows, the characterization of   DRS turns out to be more complex.}
{Each nonempty subset  $S$ of $V$ cuts $G$ into} 
a set $\mathcal P_S$ of edge-disjoint paths such that they are internally disjoint from $S$ and their union is $G$.

\begin{lemma}\label{n/2}
{Given a cycle $G=(V,E)$, let $S$ be a {nonempty} subset of  $V$.  Then $S$ is a DRS of $G$ if and only if no path in $\mathcal P_S$ has length longer than $\lceil n/2\rceil$  and at least one path in $\mathcal P_S$ has length shorter than $n/2$.\qed}
\end{lemma}
An instant corollary  reads:
The size of a minimal DRS of cycle $G$ is $2$ or $3$ when $n$ is odd, and is $3$ when $n$ is even; In particular, ${\tt dr}(G)=2$ when $n$ is odd, and ${\tt dr}(G)=3$ when $n$ is even. {These properties together with the next one lead to our algorithm for solving the MWDRS problem on cycles.}
\begin{corollary}\label{cor:vp}
{If some minimum weighted DRS has cardinality $3$, then there exists a minimum weighted DRS  of $G$ that contains vertex $v_p$.\qed}
\end{corollary}

\medskip
\hrule
\vspace{-1.5mm}\begin{algorithm}\label{alg2} {Finding minimum weighted DRS $S$ in cycle $G$.}
\vspace{1mm}\hrule
{\small
\begin{enumerate}
\vspace{-1mm}  \item $\omega\leftarrow w(v_1)$, $i[1]\leftarrow 1$, $j\leftarrow 1$, $W\leftarrow w(V)$
  \item {\bf for} $h=1$ {\bf to} $p$ {\bf do}
  \item \hspace{3mm}{\bf if} $w(v_h)<\omega$ {\bf then} $j\leftarrow j+1$, $i[j]\leftarrow h$, $\omega\leftarrow w(v_h)$
  \item {\bf end-for}
\item {\bf if} $j>1$ {\bf then} $k\leftarrow j$ {\bf  else} $k\leftarrow2$, $i[k]\leftarrow p$
\item {\bf if} $n$ is odd  {\bf then} $S\leftarrow \arg\min_{i=1}^nw(\{v_i,v_{i+p-1}\})$, $W\leftarrow w(S)$
\item {\bf for} $j=1$ {{\bf to}} $k-1$ {\bf do}
\item \hspace{3mm} {let $u_j$ be a vertex in $ \{v_h: i[j]+p\!\le\! h\!\le\! i[j+1]+p\}$ with \mbox{$w(u_j)=\min_{h=i[j]+p}^{i[j+1]+p}w(v_h)$}} 
\item \hspace{3mm} {\bf if} $w(v_p)+w(v_{i[j]})+w(u_j) <W$ {\bf then} $S\leftarrow \{v_p,v_{i[j]},u_j\}$, {$W\leftarrow w(S)$}
\item {\bf end-for}
\end{enumerate}
}
\vspace{-1.5mm}\hrule
\end{algorithm}
Note that $V_j:=\{v_h: i[j]+p\le h\le i[j+1]+p\}$, $j=1,\ldots,k-1$ induce $k-1$  internally disjoint paths in $G$. {It is thus clear that}  Algorithm \ref{alg2} runs in $O(n)$ time. The vertices indices  $1=i[1]< i[2]<\cdots< i[k]=p$ found by the algorithm satisfy  $w(v_{i[j]})=\min_{h=i[j]}^{i[j+1]-1}w(v_h)=\min_{h=1}^{i[j+1]-1}w(v_h)$ for every $j=1,\ldots,k-1$ and $w(v_{i[k]})=w(v_p)=\min_{h=1}^{p}w(v_h)$. Moreover, either $w(v_1)=w(v_p)$ and $k=2$, or $w(v_{i[j]})> w(v_{i[j+1]})$ for every $j=1,\ldots,k-1$. {These facts together with the properties mentioned above verify the correctness of the algorithm.}

\begin{theorem}\label{th:cycle}
{Algorithm \ref{alg2} finds in  $O(n)$  time a minimum weighted DRS of cycle $G$.}\qed
\end{theorem}

\subsubsection{The case of general $k$.}
{Our approach resembles  at a high level the one used by Epstein et al. \cite{elw12}.} {However, double resolvablity imposes more strict restrictions, and  requires extra care to overcome technical difficulties.}
Let $G_b=(V_b,E_b)$ be the graph obtained from $G=(V,E)$ by repeatedly deleting leaves.  {We call $G_b$ the {\em base graph} of $G$.} We reduce the MWDRS problem on $G$ to the MWDRS problem {on} $G_b$ (see Lemma \ref{lem:base}). The latter problem can be solved in polynomial time by exhaustive enumeration, {since, \bluecomment{as proved   in the sequel}, every minimal DRS of $G_b$ has cardinality at most $12(k-1)$ for $k\ge 2$}  

 Clearly, $G_b$ is connected and has minimum degree at least 2. A vertex in $V_b$ is called a {\em root} if in $G$ it is adjacent to some vertex in $V\setminus V_b$. Let $R$ denote the set of  roots. Clearly, $R\cap L=\emptyset$.  In $G_b$, we change the weights of all roots to zero, while the weights of other vertices remain the same as in $G$.
\begin{lemma}\label{lem:base}
Suppose that $S_b$ is a minimum weighted DRS of $G_b$. Then $(S_b\setminus R)\cup L$ is a minimum weighted DRS of $G$.\qed
\end{lemma}
 {Therefore,} for solving the MWDRS problem on a weighted $k$-edge-augmented tree $G$, we only need to find a minimum weighted DRS of base graph $G_b$ with the weights of all roots modified to be $0$.

 For 1-edge-augmented tree $G$, its base graph $G_b$ is a cycle,   whose minimum weighted DRS can be found in $O(n)$ time  {(recall Algorithm \ref{alg2})}. Combining this with  {Lemmas \ref{tree} and \ref{lem:base},} we have the following {linear time solvability}.
\begin{theorem}
There is an $O(n)$ time exact algorithm for solving the MWDRS problem on $k$-edge-augmented trees, for $k=0,1$, {including trees and cycles.}
\end{theorem}

In the remaining discussion for $k$-edge-augmented tree, we assume $k\ge 2$. A vertex is called a {\em branching vertex} of a graph if it has degree at least 3 in the graph. Recall that every vertex of
the base graph  $G_b=(V_b,E_b)$ has degree  at least $2$.
{It can be shown that (see Lemmas \ref{lem:path} and \ref{lem:four} in Appendix \ref{apx2})
  \begin{itemize}
  \item In $O(|E_b|)=O(n^2)$ time, $G_b$ can be decomposed into at most $ 3k-3$ {edge-disjoint paths whose ends are branching vertices of $G_b$ and internal vertices have degree 2 in $G_b$.}
\item For any minimal DRS set $S$ of $G_b$, and any path $P$ in the above path decomposition of $G_b$,    at most four vertices of $S$ are contained in $P$.
\end{itemize}
It follows that} every minimal DRS of $G_b$ contains at most $12(k-1)$ vertices. Our algorithm for finding the minimum weighted DRS of $G_b$ {examines}   all {possible}  subsets   of $V_b$ with cardinality at most $12(k-1)$ {by taking at most four vertices from each path in the path decomposition of $G_b$}; among these sets, the algorithm selects a DRS of $G_b$ {with} minimum weight. Have a table that stores the distances between each pair of vertices in $G_b$, {it takes $O(n^2)$ time to  test the double resolvability of a set.} Recalling Lemma \ref{lem:base}, we obtain the following strong polynomial time solvability for the general $k$-edge-augmented trees.
\begin{theorem}\label{th:ktree}
The {MWDRS} problem on $k$-edge-augmented trees can be solved in $O(n^{12k})$ time. \qed
\end{theorem}

\subsection{Wheels}
A  general {\em wheel}  $G=(V,E)$ on $n$ \bluecomment{$(\ge6)$} vertices  $v_1,v_2,\ldots,v_n$  is formed by the {\em hub} vertex $v_n$ and a cycle {$C=(V_c,E_c)$} over the vertices  {$v_1,v_2,\ldots,v_{n-1}$}, called {\em rim} vertices, {where the hub is adjacent to some (not necessarily all) rim vertices}. {We develop dynamic programming algorithm to solve the MWDRS problem on general wheels in $O(n^3)$ time.}

 We start  with complete wheels whose DRS has a very nice characterization that \bluecomment{is} related to the consecutive one property. A general wheel is {\em complete} if its hub is adjacent to every rim vertex.
\begin{equation}\label{1or2}
\text{{The distance in $G$ between any two vertices in $G$ is} either $1$ or $2$.}
 \end{equation}
\begin{lemma}\label{cw}
Given a complete wheel $G=(V,E)$, let $S$ be a proper nonempty subset of $V$. Then $S$ is a doubly resolving set  if and only if $S\cap V(C)$ is a  dominating set of $C$ and   any pair of rim  vertices outside $S$ has at least two neighbors in {$S\cap V(C)$}. 
\end{lemma}
\begin{proof}
If $S\cap V(C)$ is not a   dominating set of $C$, then there exists a vertex   $v_i \in V(C)- S$  such that $v_i$ is not adjacent to any vertex of $S\cap V(C)$. In  this case, $S$ cannot doubly resolve $\{v_n,v_i\}$ because for any $s_1,s_2\in S$, $d(s_1,v_n)-d(s_1,v_i)=-1=d(s_2,v_n)-d(s_2,v_i)$.

If $S\cap V(C)$ is a  dominating set of $C$ but there exist two cycle vertices  {$v_i,v_j \in V(C)-S$ such that $v_i,v_j$} are uniquely dominated by the same cycle vertex $v\in S$, then for any two vertices $s_1,s_2\in S$, $d(s_1,v_i)-d(s_1,v_j)=0=d(s_2,v_i)-d(s_2,v_j)$, saying that $S$ is not a doubly resolving set.

 Suppose that $S$ satisfies the condition stated in the lemma. We prove that $S$ can resolve every pair of vertices $x,y$ in $G$. When one of $x$ and $y$, say $x$, is a rim vertex in $S$,  
since $n\ge 6$, there exists  another rim vertex $z\in S-\{x\}$
that is not adjacent to  $x$. It follows from (\ref{1or2}) that $\{x,y\}$ is resolved by $\{x,z\}$ as $d(x,x)-d(x,z)=-2<-1\le d(y,x)-d(y,z)$. When both $x$ and $y$ are rim vertices outside $S$, there are two rim vertices $x'$ and $y'$ in $ S$ dominating $x$ and $y$, respectively. It follows that $\{x,y\}$ is resolved by $\{x',y'\}$. When one of $x$ and $y$, say $x$ is the hub, we only need consider the case of $y$ is a rim vertex outside $S$. Take rim vertices $z,z'$ from $S$ such that $z$ dominates $y$ and $z'$ does not dominate $y$. It follows that $\{x,y\}$ is resolved by $\{z,z'\}$ as $d(x,z)-d(x,z')=0<-1=d(y,z)-d(y,z')$.
\end{proof}

{The characterization in Lemma \ref{cw} can be rephrased as follows: A subset $S\subseteq V$ is a DRS of $G$ if and only if every set of three consecutive vertices on $C$ contains at least one vertex of $S$, and every set of five consecutive vertices on $C$ contains at least two vertices of $S$. This enables us to formulate the MWDRS problem on a complete wheel as an integer programming with {consecutive 1's and circular 1's constraints}, which can be solved in {$O(n^3\log^2n)$} time by Hochbaum and Levin's algorithm \cite{hl06}.} {(To the best of our knowledge, there is no such a concise way to formulate the metric dimension problem on complete wheels  as an integer programming with consecutive one matrix.)} {Moreover, it is not hard to see from the characterization that linear time efficiency can be achieved by dynamic programming approach. Furthermore, we elaborate on the idea to solve the MWDRS problem on more complex general wheels.}

\begin{theorem}\label{general}
The {MWDRS} problem on complete wheels can be solved in $O(n)$ time. The {MWDRS} problem on general wheels can be solved in  $O(n^3)$ time.\qed
\end{theorem}
\section{Conclusion}\label{sec:conclude}
{In this paper, we have established  $\Theta(\ln n)$ approximability of the {MDRS and MWDRS problems} on general graphs. There is still a gap of {$1+\ln\log_2n$} hidden in the big theta (see Theorems \ref{th:lower} and \ref{th:performance}). It deserves good research efforts to obtain even tighter upper bounds for the approximability. The $k$-edge-augmented trees, general wheels and cographs are known graph classes on which the {weighted metric dimension  problem} is polynomial time solvable. In this paper, we have extended the polynomial time solvability to the {MWDRS} problem for the first two graph classes. It would be interesting to see whether the problem on cographs {and} other graphs also admits efficient algorithms.}

\bibliography{source_location}

\appendix
\section*{Appendix}
\section{Equivalence}\label{apx:appl}
\bluecomment{We prove the equivalence between location sets and the doubly resolving sets in any graph $G=(V,E)$.}
\begin{observation}\bluecomment{Let $S\subseteq V$. Then $S$ fails to locate the diffusion source in $G$ in  some instance if and only if there exist distinct vertices $u,v\in V$ such that $S$ cannot distinguish between the case of $u$ being the source and that of $v$ being the source, i.e., $d_G(u,x)-d_G(u,y)=d_G(v,x)-d_G(v,y)$ for any $x,y\in S$; equivalently, $S$ is not a DRS of $G$.}
\end{observation}
\begin{proof} \bluecomment{Consider distinct $u,v\in V$ such that $u$ or $v$ is the information source.  Let $t_x$ denote the time when $x\in S$ receives the information.}

\bluecomment{Suppose that $d_G(u,x)-d_G(u,y)=d_G(v,x)-d_G(v,y)$ for any $x,y\in S$. Hence we have constant $c$ such that $d_G(u,x)-d_G(v,x)=c$ for any $x\in S$. It follows that the value of  $t_x$ when $u$ initiates the information at time $0$  and that when $v$ intimates the information at time $c$ are the same, saying that $S$ fails to be a location set.}

\bluecomment{Suppose $S$ is a DRS of $G$. So there exist $x',y'\in S$ such that $d_G(u,x')-d_G(u,y')\ne d_G(v,x')-d_G(v,y')$.   If $u$ is the source, then $t_{x'}-t_{y'}=d_G(u,x')-d_G(u,y')$; otherwise $v$ is the source, and  $t_{x'}-t_{y'}=d_G(v,x')-d_G(v,y')$. By checking the value  $t_{x'}-t_{y'}$ we can determine whether $u$ or $v$ is the source.}
\end{proof}
\section{Inapproximability}
Lemma \ref{lem:np1} states that there exists a polynomial time  transformation  that {transfers graph $G=(V,E)$ 
to graph $G'=(V',E')$ 
such that}   ${\tt dr}(G')\le {\tt ds}(G)+\lceil\log_2 n\rceil+3$.
\begin{proof}
 {Suppose $V=\{v_1,\ldots,v_n\}$, and let $d=\lceil\log_2 n\rceil$. Construct graph $G'=(V',E') $, as Hauptmann et al. \cite{hsv12} do, in the following way (refer to Fig. \ref{f1}):} For every $i=1,\ldots,n$,  $V'$ contains a  pair of vertices $\{v_i^0, v_i^1\}$ and $2(d+3)$   vertices $u_1^k,\cdots,u_{d+1}^k, u_a^k, u_b^k$ $(k=1,2$). {Moreover,} $V'$ contains an additional vertex $c$ which is connected to {all other  vertices}.  The edge set $E'$ is {defined as follows}: $ v_i^1v_j^1\in E'$ if and only if $v_iv_j\in E$ for $i,j=1,\ldots,n$;  $u_k^1 u_k^0\in E'$ for  all $k=1,2,\ldots,d+1,b$;  vertices $u_a^1$ and $u_a^0$ are both adjacent to all $2n$ vertices $v_i^0$, $v_i^1$, $i=1,\ldots,n$; vertices $v_j^1$ and $v_j^0$ are adjacent to $u_k^1$ and $u_k^0$ (resp. neither $u_k^1$ nor $u_k^0$)  if   the binary representation of $j$ has a $1$ (resp. 0) on the $k$-th position.  Suppose without loss of generality that $\{v_1,\ldots,v_{{\tt ds}(G)}\}$ is a minimum dominating set of $G$. In the following, we show that $S:=\{u_1^1,\cdots,u_{d+1}^1, u_a^1,u_b^1\}\cup\{v^1_1,\ldots,v^1_{{\tt ds}(G)}\}$, the RS given by Hauptmann et al. \cite{hsv12}, is actually a DRS of $G'$, which proves the lemma.

 \begin{figure}[h]
\centerline{  \includegraphics[width=8cm]{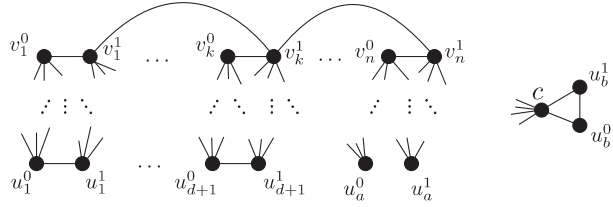}}
    \caption{{The construction of} graph $G'$}\label{f1}
\end{figure}

 {Take   arbitrary distinct $u,v\in V'$. It suffices to show that $\{u,v\}$ is doubly resolved by $S$. In view of presence of vertex $c$, the distance between any pair of vertices in $G$ is at most 2. If at least one of $u$ and $v$, say $u$, belongs to $S$, then $\{u,v\}$ is doubly resolved by $\{u,u_b^1\}$ or $\{u,u_a^1\}$} whichever has size 2. Assuming by symmetry $|\{u,u_b^1\}|=2$, we have $d_{G'}(u,u_b^1)-d_{G'}(u,u)=2$ and $d_{G'}(v,u_b^1)-d_{G'}(v,u)\le 1$. So we may assume $u\not\in S$ and $v\not\in S$.
{In case of $c\in\{u,v\}$, the pair $\{u,v\}$ is doubly resolved by $S$ as shown by the following:}
\begin{itemize}
\item The pairs $\{c,v_i^j\}$, $i\in\{1,\ldots,n\},j\in\{0,1\}$ are doubly resolved by $\{u_a^1,u_b^1\}$ as $d_{G'}(c,u_b^1)-d_{G'}(c,u_a^1)=0$ and $d_{G'}(v_i^j,u_b^1)-d_{G'}(v_i^j,u_a^1)=1$.
    \item The pairs $\{c,u_k^0\}$, $k\in\{1,\ldots,d+1\}$ are doubly resolved by $\{u_b^1,u_k^1\}$ since $d_{G'}(c,u_b^1)-d_{G'}(c,u_k^1)=0$ and $ d_{G'}(u_k^0,u_b^1)-d_{G'}(u_k^0,u_k^1)=1$.
\item The pair $\{c,u_b^0\}$ is doubly resolved by $\{u_a^1,u_b^1\}$ since  $d_{G'}(c,u_a^1)-d_{G'}(c,u_b^1)=0$ and $d_{G'}(u_b^0,u_a^1)-d_{G'}(u_b^0,u_b^1)=1$.
\item  The pair  $\{c,u_a^0\}$ is doubly resolved by $\{u_b^1,v_1^1\}$ since $d_{G'}(c,u_b^1)-d_{G'}(c,v_1^1)=0$ and $d_{G'}(u_a^0,u_b^1)-d_{G'}(u_a^0,v_1^1)=1$.
\end{itemize}
Now we may assume $c\not\in\{u,v\}$. In case of $u_k^0\in\{u,v\}$ for some $k\in\{1,\ldots,d+1,a,b\}$, the pair $\{u,v\}$ is doubly resolved by $S$ as shown by the following:
\begin{itemize}
\item The pairs   $\{u_k^0,v_i^j\}$, $i\!\in\!\{1,\ldots,n\},j\!\in\!\{0,1\}$ are doubly resolved by \mbox{$\{u_a^1,u_b^1\}$  as}  $d_{G'}(u_k^0,u_a^1)-d_{G'}(u_k^0,u_b^1)\!=\!2-d_{G'}(u_k^0,u_b^1)\ge0$, \mbox{$d_{G'}(v_i^j,u_a^1)-d_{G'}(v_i^j,u_b^1)=-1$.}
\item {The pair $\{u_k^0,u_{\ell}^0\}$ with $\ell\in\{1,\ldots,d+1,a,b\}-\{k\}$ is {doubly} resolved by $\{u_a^1,u_p^1\}$, where $p\in\{k,\ell\}-\{a\}$, since   $d_{G'}(u_p^0,u_p^1)-d_{G'}(u_p^0,u_a^1)=-1$ and $d_{G'}(u_{q}^0,u_p^1)-d_{G'}(u_{q}^0,u_a^1)=2-d_{G'}(u_{q}^0,u_a^1)\ge0$, where $\{p,q\}=\{k,\ell\}$.}
\end{itemize}
{We are left with the case where $\{u,v\}= \{v_i^j,v_{i'}^{j'}\}$ for some $i,i'\in\{1,\ldots,n\}$ and $j,j'\in\{0,1\}$. When
 $i\ne i'$, there {exists} $k\in\{1,2,\ldots,d+1\}$ such that the  binary representations  of $i$ and $i'$ {differ} at the $k$-th position (and possibly other positions). {The definition of $G'$ implies} $\{d_{G'}(v_i^j,u_k^1), d_{G'}(v_{i'}^{j'},u_k^1)\}=\{1,2\}$. It follows from
 $d_{G'}(v_i^j,u_a^1)=1= d_{G'}(v_{i'}^{j'},u_a^1)$ that  $\{v_i^j,v_{i'}^{j'}\}$ is doubly resolved by $\{u_k^1,u_a^1\}\subseteq S$. When $i=i'$, since $\{v_1,\ldots,v_{{\tt ds}(G)}\}$ is a dominating set of $G$,} there exists $h\in\{1,\ldots,{\tt ds}(G)\}$ such that $d_G(v_h,v_i)\le1$. Therefore $\{u,v\}=\{v_i^0,v_i^1\}$ is doubly resolved by $\{v^1_h,u_a^1\}$ because {$d_{G'}(v_i^0,v^1_h)-d_{G'}(v^1_i,v_h^1)=1$ and   $d_{G'}(v_i^0,u^1_a)-d_{G'}(v^1_i,u_a^1)=0$}.
\end{proof}

\section{Approximability}\label{apx1}

Lemma \ref{ich} states that
 {\em $IC(T,\mathcal T_0)\ge IC(T,\mathcal T_1)$ for any sets $\mathcal T_0$ and $\mathcal T_1$ of super tests with \mbox{$\mathcal T_0\subseteq\mathcal T_1$.}}
\begin{proof}
It has been proved in  
\cite{bdk05}  that  $IC(T,\mathcal T_0)\ge IC(T,\mathcal T_1)$ if $T$ partitions each equivalence class of $\equiv^{\mathcal T}$ into at most two equivalence classes. Suppose that the equivalence classes of $\equiv^{\mathcal T}$ are $E_1,\ldots,E_m$,   $T$ partitions each $E_i$ into $k_i$ equivalence classes, and {$k:=\max_{i=1}^mk_i-1\ge 2$. We may consider $T$   as $k$ successive tests $T_1, \ldots, T_{k}$,} each of which partitions {an} equivalence class into at most two equivalence classes. For example, assume that $E_i$ is partitioned into {$E_{i1},\ldots,E_{ik_i}$ by $T$. We consider $T_j$ ($1\le j\le k_i-1$)   partitioning $\cup_{h=j}^kE_{ih}$ into $E_{ij}$ and $\cup_{h=j+1}^{k_i}E_{ih}$, and $T_j$ ($k_i\le j\le k-1$) leaving $E_{i1},\ldots,E_{ik_i}$ unchanged.  Using {the result in} 
\cite{bdk05} we have $IC(T,\mathcal T_0)=H_{\mathcal T_0}-H_{\mathcal T_0\cup T}=(H_{\mathcal T_0}-H_{\mathcal T_0\cup T_1})+(H_{\mathcal T_0\cup T_1}-H_{\mathcal T_0\cup T_1\cup T_2})+\cdots+(H_{\mathcal T_0\cup T_1\cup \cdots\cup T_{k-1}}-H_{T_0\cup T_1\cup \cdots\cup T_{k}})\ge (H_{\mathcal T_1}-H_{\mathcal T_1\cup T_1})+(H_{\mathcal T_1\cup T_1}-H_{\mathcal T_1\cup T_1\cup T_2})+\cdots+(H_{\mathcal T_1\cup T_1\cup \cdots\cup T_{k-1}}-H_{\mathcal T_1\cup T_1\cup \cdots\cup T_k})=H_{\mathcal T_1}-H_{\mathcal T_1\cup T_1\cup \cdots\cup T_k}=H_{\mathcal T_1}-H_{\mathcal T_1\cup T}=IC(T,\mathcal T_1)$, as desired.}
\end{proof}

{To prove the approximation ratio of Algorithm \ref{alg1} in Section \ref{sec:apx} for finding minimum weighted set of super sets, we need the following lemma from \cite{bdk05}.}

\begin{lemma}[\cite{bdk05}]\label{0t1} If $IC(T,\mathcal T)>0$ then $IC(T,\mathcal T)\ge 1$.\qed
 \end{lemma}

\noindent Theorem \ref{th:performance} states that {\em Algorithm \ref{alg1} is an $O(n^3)$ time   algorithm for {the MWSTS   problem} on $(V,U_x)$ with {approximation} ratio   \[\displaystyle\ln\left( \max_{T\in U_x}IC(T,\emptyset)\right)\!+\!1\!\le\!\ln n\!+\!\ln\log_2 n\!+\!1.\]}
\begin{proof}   Suppose that an optimum solution of {the MWSTS problem on} $(V,U_x)$ is $\mathcal T^*=\{T_1^*,\ldots,T_k^*\}$. {The weight of the solution} is $\sum_{i=1}^kw(T_i^*)$. During the execution of Algorithm \ref{alg1}, for a current partial test set $\mathcal T$, let {$\mathcal T_i:=\mathcal T+T_1^*+\cdots+T_i^*$} (accordingly, $\mathcal T_0=\mathcal T$) and {$h_i:=IC(T_i^*,\mathcal T_{i-1})$.} Notice that $\sum_{i=1}^kh_i=\sum_{i=1}^k(H_{\mathcal T_{i-1}}-H_{\mathcal T{i-1}+T_i^*})=H_{\mathcal T}-H_{\mathcal T+\mathcal T^*}=H_{\mathcal T}$, since {$0\le H_{\mathcal T+\mathcal T^*}\le H_{\mathcal T^*}=0$.} Let $h_i^*<n\log_2n$ denote the initial value of $h_i$, i.e. the value of $h_i$ with $\mathcal T=\emptyset$.

During the $j$-th iteration of the while loop, Algorithm \ref{alg1} selects a test $T$ (with, say, $IC(T,\mathcal T)=\Delta_j$) and changes $\mathcal T$ {to} $\mathcal T+T$. As a result, {entropy} $H_{\mathcal T}$ drops by $\Delta_j$ and $h_i$ drops by some {$\delta_{i,j}$ such that} $\sum_{i=1}^k\delta_{i,j}=\Delta_j$. This iteration adds $w(T)$ to the solution {weight}. We distribute   {$w(T)$} among the elements of $\mathcal T^*$ by charging {each} $T_i^*$ with $w(T)\cdot\delta_{i,j}/\Delta_j$. Recall {from Lemma \ref{ich} that} $h_i=IC(T_i^*,\mathcal T_{i-1})\le IC(T_i^*,\mathcal T)$. {By the choice of $T$, we have $\frac{\Delta_j}{w(T)}=\frac{IC(T,\mathcal T)}{w(T)}\ge\frac{IC(T_i^*,\mathcal T)}{w(T_i^*)}$.} 
Therefore reducing the current $h_i$ by $\delta_{i,j}$ is associated with a charge that is at most {$w(T_i^*)\cdot\delta_{i,j}/IC(T_i^*,\mathcal T)\le w(T_i^*)\cdot\delta_{i,j}/h_i$.}

 Let $w(h)$ be the supremum of possible sums of charges that some $T_i^*$ may receive starting from the time when $h_i=h>0$. By induction on the number of such positive charges we  show $w(h)\le(1+\ln h)\cdot w(T_i^*)$. If this number is $1$, then $h>0$ and hence $\ln h\ge 0$ (by Lemma \ref{0t1}), while the charge is at most $w(T_i^*)$. In the inductive step, we consider the situation, {starting with $h_i=h$, where} 
 $T_i^*$ 
 receives a single charge {at most} $w(T_i^*)\cdot\delta/h$, $h_i$  is reduced to {$h-\delta>0$} and afterwards. Because of $h-\delta>0$, Lemma \ref{0t1} {gives} $h-\delta\ge1$. By induction assumption, $T_i^*$ receives at most $w(h-\delta)$ charges, and $w(h)\le w(h-\delta)+ \frac{\delta}{h}w(T_i^*)\le (1+\ln(h-\delta)+\frac{\delta}{h})w(T_i^*)$. {Now $w(h)\le(1+\ln h)\cdot w(T_i^*)$ follows from $1+\ln(h-\delta)+\frac{\delta}{h}<(1+\int_{1}^{h-\delta}\frac{dx}{x}+\int_{h-\delta}^h\frac{dx}{x})=1+\int_1^h\frac{dx}{x}=1+\ln h $.}
{Recalling $h_i^*<n\log_2n$, we have $w(h_i^*)<(1+\ln n+\ln\log_2n)\cdot w(T_i^*)$.} 
 This proves our result on the approximation ratio.

 To implement Algorithm \ref{alg1}, we first compute a table   storing the distances between each pair of vertices in $G$, which takes $O(n^3)$ time. Having this table, for any given 
$T\in U_x-\mathcal T$, it is easy to compute in $O(n)$ time the equivalence classes of $\equiv^{\mathcal T\cup T}$ based the equivalence classes of $\equiv^{\mathcal T}$. Besides, notice that $|U_x|=n-1$ and the while loop is executed at most $n$ times. So the algorithm runs in $O(n^3)$ time. 
 \end{proof}

\section{Cycles}\label{apx:cycle}
Let $G=v_1v_2\cdots v_nv_1$ be a cycle, {with $w(v_p)=\min_{i=1}^nw(v_i)$,} where $p:=\lceil n/2\rceil$. 

\medskip\noindent Lemma \ref{n/2} states that {\em
Let $S$ be a {nonempty} subset of  $V$.  Then $S$ is a DRS of $G$ if and only if no path in $\mathcal P_S$ has length longer than $\lceil n/2\rceil$  and at least one path in $\mathcal P_S$ has length shorter than $n/2$. }
\begin{proof}
If there exists one path in $\mathcal P_S$ is longer than $\lceil n/2\rceil$, suppose without loss of generality that this path is $v_1 v_2 \cdots v_{\lceil n/2\rceil+1}$. So from the definition of $\mathcal P_S$ we know, except the two end vertices, there is no other vertex of $S$ on the path. Then we say $v_1$ and $v_2$ can not be doubly resolved by set $S$, because for any vertex $s\in S$, $d(s,v_2)-d(s,v_1)=1$.

If there is no path in $\mathcal P_S$ has length longer than $\lceil n/2\rceil$ and no path in $\mathcal P_S$ has length shorter than $n/2$, then it means $n$ is an even number and $|S|=2$. Suppose without loss of generality, $S=\{v_1, v_{1+n/2}\}$. Then $S$ can not doubly resolve the pair of vertices $\{v_i, v_{n-i+2}\}$ for $i\neq 1, 1+n/2$. So $S$ is not a DRS.

If no path in $\mathcal P_S$ has length longer than $\lceil n/2\rceil$ and at least one path in $\mathcal P_S$ has length shorter than $n/2$, we show that for any two vertices, they can be doubly resolved by $S$. For any two vertices $u,v\in V$, if there is no vertex of $S$ on the shortest path from $u$ to $v$, then it must be the case that $u,v$ are both on one path of $\mathcal P_S$. Since no path in $\mathcal P_S$ has length longer than $\lceil n/2\rceil$, $u,v$ can be doubly resolved by the two end vertices of that path. If there are two or more vertices of $S$ on the shortest path from $u$ to $v$, say $v_i, v_j$, then $(d(u,v_i)-d(u,v_j))\cdot(d(v,v_i)-d(v,v_j))<0$, which implies that $u,v$ are doubly resolved by $\{v_i,v_j\}$. If there is only one vertex $v_i\in S$ on the shortest path from $u$ to $v$, then there is another vertex $v_j\in S$ such that $u$ is on one path of $\mathcal P_S$, whose two end vertices are $v_i,v_j$. If $\{v_i,v_j\}$ can doubly resolve $u,v$, then we are done; Otherwise, $d(u,v_i)-d(u,v_j)=d(v,v_i)-d(v,v_j)$, then it must be the case $n=d(u,v_i)+d(u,v_j)+d(v,v_i)+d(v,v_j)$, so $n$ is an even number and there exists a third vertex $v_{j'}\in S(j\neq j')$, s.t. $d(u,v_i)-d(v,v_i)\neq d(u,v_{j'})-d(v,v_{j'})$. So $S$ can still doubly resolve $\{u,v\}$.
\end{proof}
\noindent Corollary \ref{cor:vp} states that {\em If some minimum weighted DRS has cardinality $3$, then there exists a minimum weighted DRS {of $G$ that contains vertex $v_p$.}}

\begin{proof}
Suppose without loss of generality $\{v_i,v_j,v_k\}$ is a minimum weighted {DRS}, where $i<j<p<k$. Then from Lemma \ref{n/2}, {either $\{v_i,v_p,v_k\}$ or $\{v_i,v_j,v_p\}$} is a DRS. Since the weight $w(v_p)$ of $v_p$ is minimum among all vertices, $\{v_i,v_p,v_k\}$ or $\{v_i,v_j,v_p\}$ is also a minimum weighted DRS.
\end{proof}

\noindent Theorem \ref{th:cycle} states that
{\em Algorithm \ref{alg2} finds in  $O(n)$  time a minimum weighted DRS of cycle $G$.}


\begin{proof}It suffices to show that at least one of $\{v_i,v_{i+p-1}\}$, $i=1,2,\ldots,p$, and $\{ v_{i[j]},v_p,u_j\}$, $j=1,2,\ldots,k-1$ is a minimum weighted DRS of $G$. Notice from Lemma \ref{n/2} that  the doubly resolving sets of $G$ that have size two are exactly  $\{v_i,v_{i+p-1}\}$, $i=1,2,\ldots,p$,  {and all of $\{ v_{i[j]},v_p,u_j\}$, $j=1,2,\ldots,k-1$ are doubly resolving sets of $G$.} 
{It remains to} consider the case of three vertices. Suppose that $G$ has no minimum weighted DRS of size two, and {by Corollary \ref{cor:vp} that} $S=\{v_\ell,v_p,v_r\}$ is a minimum weighted DRS. By Lemma \ref{n/2} we may assume {$1\le \ell< p<1+p\le r\le n$}. 

If    $\ell=i[j]$ for some $1\le j<k$,  then Lemma \ref{n/2} implies that $r\ge i[j]+p-1$ and every $\{v_\ell,v_p,v_h\}$ with $v_h\in V_j$ is a DRS of $G$. We may assume that {$n\ge r\ge i[j+1]+p+1 $ as otherwise either $\{v_\ell,v_{i[j]+p-1}\}$ (when $r=i[j]+p-1$) or $\{v_\ell,v_p,u_j\}$ (when $r\in V_j$) would be a minimum weighted DRS of $G$. It follows  that} 
$\{ v_{i[j+1]},v_p,v_r\}$ is a DRS (by Lemma \ref{n/2}), {and thus a DRS of weight smaller than $w(S)$ as   $w(v_{i[j+1]})< w(v_{i[j]})$, a contradiction.} 

If   $i[j]<\ell<i[{j+1}]$ for some  $1\le j\le k-1$, then $w(v_{\ell})\ge w(v_{i[j]})\ge w(v_{i[j+1]})$.   {Similar to the above, we have}  $ r\ge \ell+p$. If $r\ge i[j+1]+p$, then $\{v_{i[j+1]},v_p,v_r\}$ is a DRS of $G$. The minimality of $S$ enforces that $w(v_{i[j+1]})=w(v_\ell)$ and $\{v_{i[j+1]},v_p,v_r\}$ is a minimum weighted DRS. As argued in the preceding paragraph, this set must be  $\{v_{i[j+1]},v_p,u_{j+1}\}$. Now we are left with the case $\ell+p\le r\le i[j+1]+p-1$. 
Recall that $u_j$ is {a} vertex on the path {$v_{i[j]+p}v_{i[j]+p+1}\cdots v_{i[j+1]+p}$} that has minimum weight, saying $w(u_j)\le w(v_r)$. It follows that {$\{v_{i[j]},v_p,u_j\}$ is a DRS with weight at most that of $S$. Thus}  
$\{v_{i[j]},v_p,u_j\}$  is also a minimum weighted DRS of $G$.
\end{proof}

\section{General $k$-edge-augmented trees}\label{apx2}
Observe that $G[R\cup (V\setminus V_b)]$ is a forest, where every component is a tree rooted at some unique root in $R$. Let $T_r$ denote the tree (component) rooted at  $r\in R$. In $G_b$, we change the weights of all roots to zero, while the weights of other vertices remain the same as in $G$. Next we present the proof of Lemma \ref{lem:base}. 

\medskip\noindent
  Lemma \ref{lem:base} states that {\em
Suppose that $S_b$ is a minimum weighted DRS of $G_b$. Then $S= (S_b\setminus R)\cup L$ is a minimum weighted DRS of $G$.}
\begin{proof}
Using the argument in the proof of Lemma \ref{tree}, we can easily prove $L$ is contained in every 
DRS of $G$.  Next we prove that $S$ is a   DRS of $G$.

Consider any two vertices   $u,v\in R\cup(V\setminus V_b)$. If $|L|\ge 2$, then $u,v$ must be on some shortest path between two leaves of $G$, 
which implies that 
$u,v$ are doubly resolved by these two leaves in $L$.
 If $|L|=1$, then $u,v$ are on the path from the only leaf, denoted as $l$, to the only root, denoted as $r$. Take 
 $s\in S_b\setminus \{r\}$. Since $d(l,u)-d(l,v)\neq d(r,u)-d(r,v)=d(s,u)-d(s,r)-(d(s,v)-d(s,r))=d(s,u)-d(s,v)$, it follows that $u,v$ are {doubly resolved} by $\{l,s\}$  ($\subseteq S=(S_b\setminus R)\cup L$). 

Consider any $u,v\in V_b$. Since $S_b$ is a DRS of $G_b$ and $d_G(u,v)=d_{G_b}(u,v)$, we see that $u,v$ are doubly resolved by $S_b$ in $G$.  If $R\cap S_b=\emptyset$, then $u,v$ are doubly resolved by $S$. Otherwise, we take 
 $r\in R\cap S_b$ and $l_r\in L\cap V(T_r)$. Since  $d(l_r,u)- d(l_r,v)=d(r,u)-d(r,v)$, it follows that $u,v$ are also doubly resolved by $S$. 

 Consider any $u\in V_b\setminus R$ and $v \in V\setminus V_b$. Suppose that $v$ is a vertex of tree $T_r$, where $r\in R$.
 There exists a leaf  $l\in L\cap V(T_r)$ such that $d(r,l)=d(r,v)+d(v,l)$. Since $r,u\in V_b$, we have shown in the above that $r$ and $u$ are {doubly resolved} by $S$. It follows that there exists $s'\in S$ 
 such that $d(u,s')<d(u,r)+d(r,s')$, as otherwise  $d(u,s')-d(r,s')=d(u,r)$ for any $s'\in S$ implies a contraction. From the inequality, we deduce that $s'$ is not a vertex of tree $T_r$, saying $d(v,s')=d(v,r)+d(r,s')$. It follows that 
\[d(u,s')-d(v,s')<d(u,r)+d(r,s')-(d(v,r)+d(r,s'))=d(u,r)-d(v,r).\]
On the other hand,
\begin{eqnarray*}
d(u,l)-d(v,l)&=&d(u,r)+d(r,l)-d(v,l)\\
&=&d(u,r)+d(r,v)+d(v,l)-d(v,l)=d(u,r)+d(v,r).\end{eqnarray*}
 Hence $d(u,s')-d(v,s')<d(u,l)-d(v,l)$, saying that $u,v$ are doubly resolved by $\{s',l\}\subseteq S$. We have shown that any pair of vertices $\{u,v\}$ in $G$ can be doubly resolved by $S$. Thus $S$ is indeed a DRS of $G$.

 Now we prove the optimality of $S$. 
 Suppose on the contrary that $S'=K\cup S_b'$  is a DRS of $G$ and its weight is smaller than that of $S$, where $K\subseteq V-V_b$ and  $S_b'\subseteq V_b$. As mentioned at the beginning of  the proof,
 $L\subseteq K$, which implies that the weight of $S_b'$ is smaller than that of $S_b$. For any $w\in K$ and any two vertices $u,v\in V_b$, assume $w\in T_r$ for some $r\in R$. We have $d(u,w)-d(v,w)=d(u,r)-d(v,r)$. It follows that $R\cup S_b'$ is a DRS of   $G_b $. Note that the weights of vertices in $R$ are $0$ in $G_b$, so the weight of $R\cup S_b'$ is  smaller than that of $S_b$. It is a contradiction to the minimality of $S_b$ in $G_b$.
\end{proof}

Recall that every vertex of
  $G_b=(V_b,E_b)$ has degree  at least $2$.    {The following structural property of the base graph has been stated in \cite{elw12} with a partial proof. We give a full proof for completeness.}
\begin{lemma}\label{lem:path}
The base graph $G_b$ is decomposed into $q\leq 3k-3$ edge disjoint paths,   each of which is a minimal path in $G_b$ with both ends being branching vertices  {of $G_b$}.
\end{lemma}
\begin{proof}Observe that $G_b$ is also a $k$-edge-augmented tree $G$. Therefore $|E_b|=|V_b|-1+k$, $\sum_{v\in V_b}d_{G_b}(v)=2|V_b|+2k-2$ and $\sum_{v\in V_b}(d_{G_b}(v)-2)=2k-2$. It follows that $\sum_{v\in V_b,d_{G_b}(v)\ge3}d_{G_b}(v)\le3(2k-2)$. Note that the degree sum of branching vertices is exactly $2q$. Thus $2q\le3(2k-2)$ proves the result. \end{proof}

In the metric dimension problem, for any minimal RS $S'$, and any path $P$ in the above path decomposition of $G_b$, it was shown in \cite{elw12} that  the number of vertices in $S'$ that are  {``associated''} with $P$ is at most six. Next, we prove an analogue for  minimal doubly resolving sets. 
\begin{lemma}\label{lem:four}
Let $S$ be a minimal {DRS of $G_b$}, and let $P$ be a path in the path decomposition of $G_b$ {stated in Lemma \ref{lem:path}}. Then {$P$ contains at most four vertices from $S$.}
\end{lemma}

\begin{proof}
 Assume on the contrary that there are more than four vertices of $S$ on path $P$.  {Note that for any two vertices in $G_b$, the distance $d_{G_b}(u,v)$ between them is $d(u,v)$ the same as that in $G$.} We suppose the two outmost vertices of $S$ on $P$ are $s, t$ and the subpath of $P$ from $s$ to $t$ is  $v_0v_1\cdots v_l$ with $v_0=s$ and $v_l=t$.   {So $d_P(v_i,v_j)=|j-i|$ for any $i,j\in\{0.1,\ldots,l\}$.}
 Define \begin{center}
 $i_0:=\max \{i:v_i\in S$ and  $0\le i\le l/2\}$, $ j_0:=\min \{i:v_i\in S$ and $ l/2\le i\le l\}$.
 \end{center}
 Obviously,  $i_0\le l/2\le  j_0$,  and  {by assumption} $T:=(S-V(P))\cup\{s,v_{i_0},v_{j_0},t\}$ is a proper set of $S$. The minimality of $S$ says that $T$ can not doubly resolve some pair $R$ of vertices of $G_b$. Write $U:=\{v_i:0<i<l\}$. We distinguish among three cases depending on the value of $|R\cap U|$, which is $0$ or $1$ or $2$. 

 \medskip
\noindent {\em Case 1.} $|R\cap U|=0$.
Suppose that $R=\{u,v\}$, and $\{s_1,s_2\}$ doubly resolves $R$ with $s_1,s_2\in S$ and $s_1\not\in T$.  {If there exists $s_h\in\{s_1,s_2\}- T $ such that by switching $u$ and $v$ if necessary we have   $d(u,s_h)=d(u,s)+d_P(s,s_h)$, $d(v,s_h)<d(v,s)+d_P(s,s_h)$   and $d(v,s_h)=d(v,t)+d_P(t,s_h)$, $d(u,s_h)<d(u,t)+d_P(t,s_h)$, then these four (in)equalities imply $d(u,s)-d(v,s)<d(u,s_h)-d(v,s_h)<d(u,t)-d(v,t)$, saying $R$ is doubly resolved by $\{s,t\}$. The contradiction shows that for any $s_h\in\{s_1,s_2\}- T$, there exists $r_h\in\{s,t\}$ such that $d(u,s_h)=d(u,r_h)+d_P(r_h,s_h)$ and $d(v,s_h)=d(v,r_h)+d_P(r_h,s_h)$, giving
\[d(u,s_h)-d(v,s_h)=d(u,r_h)-d(v,r_h)\text{ for every }s_h\in\{s_1,s_2\}- T.
\]
Since $R$ is doubly resolved by $\{s_1,s_2\}$, it is easy to see from the above equation that $R$ is doubly resolved by $\{r_1,r_2\}$ if $s_2\not\in T$ and by $\{r_1,s_2\}$ otherwise, a contradiction. }

\medskip
\noindent{\em Case 2.} $|R\cap U|=1$. Suppose $R=\{u,v_j\}$ with  $u\in V-U$ {and $v_j\in U$. Note that $1\le j\le l-1$.} 
In case of $j<i_0$,    note from $i\le l/2$ that  {$ d_P(v_j,v_{i_0})=i_0-j<l/2$. If $d(u,v_{i_0})=d(u,s)+d_P(s,v_{i_0})=d(u,s)+i_0$}, then $R$ is doubly resolved by 
$\{s,v_{i_0}\}$ since  {$d(u,v_{i_0})-d(u,s) =i_0>d_P(v_j,v_{i_0})\ge d(v_j,v_{i_0})-d(v_j,s)$}; otherwise {$d(u,v_{i_0})=d(u,t)+d_P(t,v_{i_0})=d(u,t)+l-i_0$}, and $R$ is 
doubly resolved by $\{v_{i_0},t\}$ since  {$d(u,v_{i_0})-d(u,t)=l-i_0\ge l/2>d_P(v_j,v_{i_0})\ge d(v_j,v_{i_0})-d(v_j,t)$}. Similarly, in case of $j_0<j$, we obtain the contradiction that $R$ is doubly resolved by either $\{t,v_{j_0}\}$ or $\{v_{j_0},s\}$.

In case of  $0<i_0\le j\le j_0<l$,  {let $x\in\{s,t\} $ and $y\in\{v_{i_0},v_{j_0}\}$ satisfy $d(u,x)=\min\{d(u,s),d(u,t)\}$ and $d(v_j,y)=\min\{d(v_j,v_{i_0}),d(v_j,v_{j_0})\}$. Observe that $\min\{d(u,s),d(u,t)\}<\min\{d(u,v_{i_0}),d(u,v_{j_0})\}$ and $ \min\{d(v_j,v_{i_0}),d(v_j,v_{j_0})\}$ $ <\min\{d(v_j,s),d(v_j,t)\}$.}
It follows that $d(u,x)-d(u,y)<0<d(v_jx)-d(v_j,y)$,  contradicting the fact that $\{x,y\}\subseteq\{s,t,v_{i_0},v_{j_0}\}\subseteq T$ does not doubly resolve $R$. 
 Notice from  $|S\cap V(P)|>4$ that $\{i_0,j_0\}\ne\{0,l\}$. It remains to consider the case of $0=i_0< j\le j_0<l$ and that of $0<i_0\le j<j_0=l$. By symmetry, it suffices to consider $0=i_0< j\le j_0<l$.

  If  {$d(v_j,v_{j_0})=d(v_j, t)+d_P(t,v_{j_0})$},  {then from the definitions of $i_0,j_0$ it is easy to see that} 
  $d(v_j,v)-d(s,v)=j$ for every $v\in S$, saying that $S$ does not doubly resolve $\{v_j,s\}$. The contradiction enforces
   {\[ d(v_j, t)+d_P(t,v_{j_0})>d(v_j,v_{j_0}).\]} If  {$d(u,v_{j_0})=d(u,s)+d_P(s,v_{j_0})$, then $d(u,v_{j_0})-d(u,s)=d_P(s,v_{j_0})>d_P(v_j,v_{j_0})$ $\ge d(v_j,v_{j_0})-d(v_j,s) $} shows that $R$ is doubly resolved by $\{v_{j_0},s\}$, a contradiction. Therefore we have  {$d(u,v_{j_0})=d(u,t)+d_P(t,v_{j_0})$}.  Since $\{v_{j_0},t\}$ does not doubly resolve $R=\{u,v_j\}$, we have  {$d_P(t,v_{j_0})=d(v_j,v_{j_0})-d(v_j,t)$}, a contraction to  {$d(v_j, t)+d_P(t,v_{j_0})>d(v_j,v_{j_0})$}.  

\medskip
\noindent {{\em Case 3.} $|R\cap U|=2$. Suppose $R=\{v_i,v_j\}$ with $0< i<j< l$. If $j\le i_0$,} 
then $R$  is doubly resolved by $\{s,v_{i_0}\}$ since  {$i_0\le l/2$ implies $d(s,v_i)-d(s,v_j)=d_P(s,v_i)-d_P(s,v_j)=i-j$ and $d(v_{i_0},v_i)-d(v_{i_0},v_j)=d_P(v_{i_0},v_i)-d_P(v_{i_0},v_j)=j-i$.} {Similarly, if $i\ge j_0$, then $R$ would be doubly resolved by $\{t,v_{j_0}\}$. Therefore
\[j>i_0\text{ and }i<j_0.\]}
If $i<i_0<j\le j_0$,  {then the shortest path between $v_j$ and $s$ in $G_b$ contains either $v_{i_0}$ or $t$. In the former case, $R$ is} 
doubly resolved by $\{s,v_{i_0}\}$ as  {$d(v_i,s)-d(v_i,v_{i_0})=d_P(v_i,s)-d_P(v_i,v_{i_0})=i-(i_0-i)<i_0$ and $ d(v_j,s)-d(v_j,v_{i_0})=d_P(v_j,s)-d_P(v_j,v_{i_0})=i_0$};  {in the latter case, $R$ is} 
 doubly resolved by $\{s,t\}$ because $d(v_i,s)-d(v_i,t)<0$ while $d(v_j,s)-d(v_j,t)>0$.   {Similarly, if $i_0\le i\le j_0<j$, then $R$ would be doubly resolved by $\{t,v_{j_0}\}$ or $\{t,s\}$. Hence it must be the case that
 \[\text{either }i_0\le i<j\le j_0\text{ or }i<i_0\le j_0<j.\]}
If $i<i_0\le j_0<j$, then $R$ is doubly resolved by $\{s,t\}$ because   {$d(v_i,s)=d_P(v_i,s)\le l/2$ and $d(v_j,t)=d_P(v_j,t)\le l/2$   imply} $d(v_i,s)-d(v_i,t)<0<d(v_j,s)-d(v_j,t)$. Now we are   left with the case of $i_0\le i<j\le j_0$.

 Since $\{v_{i_0},v_{j_0}\}$ does not doubly resolve $R$, there exist   {$x\in R=\{v_i,v_j\}$} and $y\in \{v_{i_0},v_{j_0}\}$ such that some shortest path, denoted as $Q$,  between $x$ and $y$ goes through $s$ and $t$.  If   {$x=v_i$ and $y=v_{i_0}$}, then path $Q$ goes through  $v_iv_{i+1}\cdots v_j\cdots $ $v_{j_0}v_{j_0+1}\cdots v_{l-1}v_l $; it follows from the definitions of $i_0$ and $j_0$ that all shortest paths between $v_i$ and vertices in $S$ go through $v_j$, implying $d(v_i,v)-d(v_j,v)=j-i$ for every $v\in S$, a contradiction to  resolvability of $S$. Therefore $\{x,y\}\ne\{v_i,v_{i_0}\}$, and similarly $\{x,y\}\ne\{v_j,v_{j_0}\}$. This particularly says that $G_b$ has only one shortest path between $v_i$ and $v_{i_0}$ (resp. $v_j$ and $v_{j_0}$), which is the one contained in $P$. Thus
 \[d(v_i,s)+d(s,v_{i_0})>d(v_i,v_{i_0}) \text{ and }d(v_j ,t)+d(t,v_{j_0})>d(v_j,v_{j_0}).\]
 In view that $\{x,y\}$ is either $\{v_i,v_{j_0}\}$ or $\{v_j,v_{i_0}\}$, we assume by symmetry that $x=v_i$ and $y=v_{j_0}$. From  {the shortest path} $Q=v_iv_{i-1}\cdots v_0\cdots v_lv_{l-1}\ldots v_{j_0}$ we deduce that $d(v_i,v_{j_0})-d(v_i,t)=d(t,v_{j_0})$.   Since $\{v_{j_0},t\}$ does not doubly resolve $R=\{v_i,v_j\}$, we have  $d(t,v_{j_0})=d(v_j,v_{j_0})-d(v_j,t)$. However this shows a contradiction to $d(v_j ,t)+d(t,v_{j_0})>d(t,v_{j_0})$.  

 \medskip
 The  {contradictions} derived in the above cases prove the lemma.
 \end{proof}

A corollary of Lemmas  \ref{lem:base},   \ref{lem:path} and Lemma \ref{lem:four} gives $|L|\le {\tt dr}(G)\le |L|+12(k-1)$ for $k$-edge-augmented tree $G$. On the other hand, it was shown in \cite{elw12} that ${\tt md}(G_b)\le18(k-1)$. Let graph $G'$ be obtained from $G_b$ by attaching a pendant edge to each vertex. Then $G'$ has exactly $|V_b|$ leaves, and   ${\tt dr}(G')\ge|V_b|$. Using Lemma 3 of \cite{elw12}, one has {${\tt md}(G')\le18(k-1)$,} and ${\tt dr}(G')/{\tt md}(G')=\Omega(|V_b|)$ can be arbitrarily large.
\section{Wheels} \label{apx:wheel}
A   {\em wheel}  $G=(V,E)$ on $n$ vertices $1,2,\ldots,n$ is formed by the {\em hub} vertex $n$ and a cycle {$C=(V_c,E_c)$} over the vertices $1,2,\ldots,n-1$, called {\em rim} vertices, {where the hub is adjacent to some (not necessarily all) rim vertices}. The neighbors of hub $n$ are called \emph{connectors}. {The dynamic programming approach has been proved successful in finding minimum weighted resolving sets on wheels \cite{elw12}. In this section, we show that the approach also works for DRS: a minimum weighted DRS on a wheel can be found in cubic time; the computing time improves to be linear when the wheel is complete.}

 For simplicity, in this  {section} we consider $G=(V,E)$ a (complete or uncomplete) wheel with at least 13 connectors, and {design exact algorithms for solving the MWDRS problem on  $G$ in cubic time. Observe that wheels with at most 12 connectors are $k$-edge-augmented trees with $k\le 12$ in which the MWDRS  problem is solvable in strongly polynomial time (see Theorem \ref{th:ktree}).}

The clockwise order of rim vertices along $C$ is $1,2,\ldots,n-1,1$. Given any rim vertices $u$ and $v$ (possibly $u=v$), we use $C[u,v]$ to denote the clockwise path of $C$ from $u$ to $v$.  {In case of $u=v$, path $C[u,v]$ consists of a single vertex $u$.} Define $C(u,v):=C[u,v]-\{u,v\}$.    Let $u$ and $v$ (possibly $u=v$) be rim vertices. We say that $u$ and $v$ are {\em close} to each other if $d(u,v)<d(u,n)+d(n,v)$.  We say that rim vertex $u$ is close to vertex subset $S$ if there is $s\in S$ such that $u$ is close to $s$.

Suppose that rim vertices $u$ and $v$ are close to each other. We say that $u$ is close to $v$ {\em from the left}, equivalently $v$ is close to $u$ {\em from the right}, if $C[u,v]$ has length $d(u,v)$. {Clearly, a rim vertex is close to itself from both the left and the right.} Observe that   {\em any two vertices on $C[u,v]$ are closed to each other}.
   {This fact and the following observation are frequently used in our discussion   implicitly or explicitly.}
\begin{Observation}\label{ob1}
 Let $s\in V_c $. The set of vertices which are close to $s$ induces a path $P_s$ of $C$ containing $s$. Moreover, suppose $P_s=C[u,v]$ is the clockwise path from $u$ to $v$. Then $C[u,s]$
(resp. $C[s,v]$) contains 
at most two connectors. \qed
\end{Observation}

Consider $S\subset V$.  {If some rim vertex is not close to $S$, then  $d(u,s)=d(u,n)+d(n,s)$ for any $s\in S$, which implies that} $S$ can not doubly resolve $\{u,n\}$.  Consider rim vertices {$u,v,s$. If} $s$ is close to neither $u$ nor $v$, 
then $d(u,s)-d(v,s)=d(u,n)-d(v,n)$. Apparently we have the following observation.

\begin{Observation}\label{ob2}
Let   {$S$ be a DRS of $G$. The following}  hold:
\begin{mylabel}
 \item Every rim vertex  is close to $S$.
  \item For every $s\in S$, if there exist rim vertices $u,v$ close to $s$ such that
$d(s,u)-d(s,v)=d(n,u)-d(n,v)$, then there is  $s'\in S-\{s,n\}$ which is close to at least one of  $u$ and $v$
such that $d(s,u)-d(s,v)\neq d(s',u)-d(s',v)$.\qed
 \end{mylabel}
\end{Observation}

In our discussion, the additions and subtractions over {the vertex indices} $1,2, \ldots,n-1$ and $\infty$ are taken module $n-1$. The results of the additions and subtractions are numbers in $\{1,2,\ldots,n-1\}$ or $\infty$ or $-\infty$, where module operation on $\infty$ (resp. $-\infty$) always produces $\infty$ (resp. $-\infty$).
\begin{lemma}\label{ob3}
Suppose $S\subseteq V$ is a set  {to which every rim vertex is close. Then for any rim vertices $u,v$ (possibly $u=v$)} that close to a common vertex $s\in S$, there exists   {$s'\in S-\{s\}$ which is close to neither $u$ nor $v$.}
\end{lemma}
\begin{proof}
By  Observation \ref{ob1}, there are at most 4 connectors contained in  $C[u,v]$. So there are at least $9$ connectors on  {$C(v,u)$, and we can take connector $t$ from $C(v,u)$ such that each of $C(v,t]$ and $C[t,u)$ contains at least 5 connectors. It follows from Observation \ref{ob1} that $P_t=C[x,y]$ is a proper subpath of $D(v,u)$, each of $C(v,x)$ and $C(y,u)$ contains at least 3 connectors, and $P_u\cup P_v$ is vertex-disjoint from $P_t$. Suppose that $t$ is close to $s'\in S$. Then $s'$ is contained in $P_t$, and thus is outside $P_u\cup P_v$, saying that $s$ is close to neither $u$ nor $v$.}  
\end{proof}

\begin{lemma}\label{dis}
Suppose two rim vertices $u,v$ are both close to rim vertex $s$ and $d(u,s)-d(v,s)=d(u,n)-d(v,n)$. If $u$ is  close to rim vertex $s' $ $(\neq s)$, then $u,v$ can be doubly resolved by $\{s,s'\}$.
\end{lemma}
\begin{proof}
If $v$ is close to $s'$, then  {$P_s$ and $P_{s'}$ both contain $u,v$}, and it is easy to see that $u,v$ {are doubly} resolved by $\{s,s'\}$. If $v$ is not close to $s'$, then $d(v,s')=d(v,n)+d(n,s')$. Using $d(u,s')<d(u,n)+d(n,s')$, we derive $d(u,s')-d(v,s')<d(u,n)-d(v,n)=d(u,s)-d(v,s)$, showing that $u,v$  are doubly resolved by $\{s,s'\}$.
\end{proof}

Let $S$ $(\subseteq V)$ be a set of observers. Given two distinct rim observers $s,s'\in S$, we say that $s$ and $s'$ are {\em consecutive} if $C(s',s)$ or $C(s,s')$ contains no observer. In the former  case, we say that $s'$ is {\em on the left}  of $s$ and  write $s'=s^-$; in the latter case,   we say that $s'$ is {\em on the   right}  of $s$ and  write $s'=s^+$. Thus the observers contained in $C[s^-,s^+]$ are exactly $s^-$, $s $ and $s^+$.

 Given a rim observer $s$, we call a pair of rim vertices $x,y$   a \emph{bad pair}   of $s$, if $x,y$ are both close to $s$, one from the left and one from the right,  $d(x,s)=d(y,s)$, and {$d(x,n)=d(y,n)$}.  A bad pair $x,y$ of $s$ is called {\em minimal} if $d(x,s)=d(y,s)$ is minimum 
 among all bad pairs of $s$. We say a minimal bad pair $x,y$ of $s$ is {\em covered}   by $s^-$ from the left  if $s^-$ is closed to at least one of $x$ and $y$, and {\em covered}  by $s^+$ from the right
 if $s^+$ is closed to at least one of $x$ and $y$.    We say the minimal bad pair of $s$ is {\em covered} by $S$ if it is covered by $s^-$ or $s^+$.

{Let $S$ ($\subseteq V$) be an observer set  and  $s\in S$ be a rim vertex. We say that $S$ is {\em left-continuous at $s$} if for any rim vertex $x\in C[s^-,s)$ that is close to $s$ from the left and satisfies $1=d(x,s)-d(x+1,s)=d(x,n)-d(x+1,n)$, we have $x$  close to $s^-$   from the right. We say that $S$ is {\em right-continuous at $s$} if for any rim vertex $x\in C(s, s^+]$ that is close to $s$ from the right and satisfies $1=d(x,s)-d(x-1,s)=d(x,n)-d(x-1,n)$, we have $x$  close to $s^+$   from the left.}

\begin{theorem}\label{nwhl}
An observer set $S\subseteq V$ is a {DRS of $G$} if and only if the following three conditions are satisfied.
\begin{mylabel}
\item Every rim vertex   is close to $S$.
\item {The minimal bad pair of any  vertex in $S-\{n\}$ (if any) is covered by $S$.} 
\item {$S$ is left-continuous and right-continuous at every vertex of $S-\{n\}$.}
\end{mylabel}
\end{theorem}
\begin{proof}
To see the ``only if'' part, suppose that $S$ is a {DRS} of $G$.  Observations \ref{ob2}(i) gives (i) immediately. If $\{u,v\}$ is the minimal bad pair  of some $s\in S-\{n\}$, then the condition of Observation  \ref{ob2}(ii) is satisfied with $d(s,u)-d(s,v)=d(n,u)-d(n,v)=0$, and we have one of $u$ and $v$   close to some $s'\in S-\{s,n\}$.  Using Observation \ref{ob1}, it is easy to see that one of $u$ and $v$ is close to $s^-$ or $s^+$, establishing (ii). To show {the left-continuous of $ S-\{n\}$,  suppose on the contrary that there exists $s\in S-\{n\}$ such that} some $x\in C[s^-,s)$  is close to $s$ from the left,  $1=d(x,s)-d(x+1,s)=d(x,n)-d(x+1,n)$ and $x$ is not close to $s^-$. Then for any $s'\in S$, {either $x$ is not close to $s'$ or $C[x,s']\supseteq C[x,s]$, giving} $d(x,s')=d(x,n)+d(n,s')$ or $d(x,s')=d(x,s)+d(s,s')$. Since $d(x,n)=1+d(x+1,n)$, in either case, some shortest path from $x$ to $s'$ starts with $x(x+1)$, giving $d(x,s')=1+d(x+1,s')$, {a contradiction to the fact that $S$ doubly resolves $\{x,x+1\}$.} 
Hence $S$ is left-continuous at every $s\in S-\{n\}$. The right-continuity is proved by a symmetric argument. So (iii) holds.

To prove the ``if '' part, suppose  {$S\subseteq V$} satisfies  conditions (i) -- (iii). Consider an arbitrary pair of vertices $\{u,v\}$. We prove that it can be doubly resolved by $S$. If $v=n$, then $u$ is a rim vertex, and there must exist $s_1,s_2\in S$ such that $u$ is close to $s_1$ (by (i)) but not close to $s_2$  (by Lemma \ref{ob3});  {it is easily checked that} $u,v$ are doubly resolved by $\{s_1, s_2\}$. If $u,v$ are   rim vertices and are not close to any common vertex in $S$, then by (i) there exist distinct rim vertices $s_1,s_2\in S$ such that $u$ is close to $s_1$  {but not to $s_2$} and $v$ is close to $s_2$  {but not to $s_1$, which implies $d(u,s_1)-d(v,s_1)<d(u,n)-d(v,n)<d(u,s_2)-d(v,s_2)$, and thus  $\{u,v\}$ is} doubly resolved by $s_1, s_2$. It remains to consider   $u,v$ being both close to some $s\in S-\{n\}$.  {Switching $u$ and $v$ if necessary, we may assume $C[u,v]\subseteq P_s$}.

 If $d(u,s)-d(v,s)\neq d(u,n)-d(v,n)$, then by Lemma \ref{ob3},   there exists $s'\in S-\{n\}$ such that neither $u$ nor $v$ is close to $s'$, which means $d(u,s')-d(v,s')=d(u,n)-d(v,n)$, and thus $\{u,v\}$  is  doubly resolved by $\{s,s'\}$.  {So we assume
 \begin{gather}\label{asp}
 d(u,s)-d(v,s)=d(u,n)-d(v,n),\end{gather}
and distinguish among  the following four cases.}

\medskip
\noindent {\em Case 1.}  $v\in C[u,s]$, {i.e., $u$ and $v$ are close to $s$ from the left.} Notice from (\ref{asp}) that 
$C[u,v-1]$ contains no connector. If   $C(u,s)$ contains some $s'\in S$, then   $u,v$  {are close to $s'$, and } can be doubly resolved by $\{s,s'\}$. Otherwise, {$u,v\in C(s^-,s]$}. 
Note that vertex $v-1$ on $C[u,v)\subseteq P_s$ is also close to $s$, and $1=d(v-1,s)-d(v,s)=d(v-1,n)-d(v,n)$. {Since $S$ is left-continuous at $s$ by (iii),} 
 vertex   $v-1$ is close to $s^-$  {from the right}, which implies that $u$ is  close to $s^-$. It follows from Lemma \ref{dis} that $\{u,v\}$ is doubly resolved by $\{s,s^-\}$.

\medskip
\noindent {\em Case 2.}  $u\in C[s,v]$, {i.e., $u$ and $v$ are close to $s$ from the left.}  
The proof is similar to that in Case 1, {using the right-continuity of $S-\{n\}$}.

\medskip
\noindent {\em Case 3.}   {$s\in C[u,v]$ and $d(u,s)=d(s,v)$. Note that $u$ and $v$ form a bad pair of $s$.}
From condition (ii), there exists $s'\in S-\{n\}$ which covers the minimal bad pair of $s$. It is easy to see that $s'$ is  close to $u$ or $v$. In turn Lemma \ref{dis} says that $\{u,v\}$ is doubly resolved by $\{s,s'\}$.

\medskip
\noindent {\em Case 4.}   {$s\in C[u,v]$ and $d(u,s)\ne d(s,v)$}. 
Suppose without loss of generality that $u<s<v$ and $1\le d(u,s)-d(v,s)=2s-u-v$.  {It is instant from (\ref{asp}) that $d(u,n)\ge2$ and $un\not\in E$.}
 {Recall that every vertex in $C[u,s]$ is close to $u$ and every vertex in $C[s,v]$ is close to $v$. In view of (\ref{asp}) and Lemma \ref{dis}, we only} need to consider the case of  \begin{center}
$C[u,v]\cap S=\{s\}$, i.e., $s^-,s^+\in C(v,u)$.
\end{center}

If there {exists}  a vertex $x$ on $C[u,s)\subseteq P_s$ such that $d(x,n)-d(x+1,n)=1$, then $x$ is {closed to $s$} 
from the left, and $1=d(x,s)-d(x+1,s)=d(x,n)-d(x+1,n) $. {Since $S$ is left-continuous at $s$, vertex} 
$x$ is close to $s^-$ from the right.  It follows that $u$ is  close to $s^-$. 
Thus lemma \ref{dis}  says that $\{u,v\}$ is doubly resolved by $\{s,s^-\}$.


 In the remaining proof, we assume that there does not exist such {an} $x$ on $C[u,s)$, {implying $d(x,n)\le d(x+1,n)$ for every $x\in C[u,s)$}. Particularly, $d(s-1,n)\le d(s,n)$. Then  either all vertices on $C[u,v)$ are connectors or none of them is a connector. Since $un\not\in E$,  {we see that} $C[u,s)$ contains no connector, and $d(s-1,n)=d(s-2,n)+1=\cdots =d(u,n)+s-1-u$. So $d(u,n)+s-1-u\le d(s,n)\le d(s,v)+d(v,n)=v-s+d(v,n)$, yielding $d(u,n)-d(v,n)+2s-u-v\le 1$. {It follows from (\ref{asp}) and $d(u,s)-d(v,s)=2s-u-v$ that} $2(d(u,s)-d(v,s))\le 1$, a contradiction to   $d(u,s)-d(v,s)\ge 1$.
\end{proof}

From  {Theorem \ref{nwhl},} 
we have the following corollary.
\begin{corollary}
If $S\subseteq V$ is a minimal DRS, then $n\not\in S$.\qed
\end{corollary}

{There are two major different features between  the characterizations of  minimal doubly resolving sets and minimal resolving sets on wheels. First,   some consecutive rim vertices are allowed to be not close to the RS. So when designing the dynamic programming algorithm for finding minimum weighted {RS}, Epstein et al. \cite{elw12} could  choose a feasible pair of start   and end vertices, without considering the stop condition. Second, to find a feasible RS, it suffices to consider the bad pair case of an observer vertex, while  {to find a feasible DRS, we additionally need to make sure that the set under consideration is left-continuous and right-continuous at every observer vertex.} 
{Both   differences make the algorithm design for the MWDRS problem  more complex and difficult.}

Next we design a strongly polynomial time   dynamic programming algorithm to find a minimum weighted DRS of $G$.  {Let $a$ be a rim vertex such that $P_a$ is shortest among all rim vertices}. Note that every minimum weighted DRS must contain some vertex from {$P_a$}. By examining all vertices in  {$P_a$},  without loss of generality we assume that vertex  {$1\in V(P_a)$} is an observer in a minimum weighted DRS.  {Our aim is to construct an observer set $S$ of minimum weight on the condition that  $1\in S$ and $S$ satisfies conditions (i)--(iv) in Theorem \ref{nwhl}. Such an $S$ must be a minimum weighted DRS of $G$ as guaranteed by Theorem \ref{nwhl}. Henceforth we assume $1\in S$.
 }

 {Before describing {our} algorithm, we introduce some notations. Let $l_1$ and $r_1$ to be the rim vertices such that $C[l_1,r_1]=P_1$ is the path induced by all vertices close to vertex $1$.  For any rim vertex $u$ with  {$1<u\le n-1$}, suppose we put  $u$ into $S$. To assure validity of conditions (ii)-(iv) of Theorem \ref{nwhl}, basically we need check every vertex in $P_u$. But by virtue of $1\in S$, we can narrow the range $P_u$ by cutting off its intersection with $P_1$. Let  {$l_u\ge1$} be the minimum index and $r_u\le n-1$ be the maximum index such that $C[l_u,r_u]$} is a subpath of $P_u$ (so all vertices on $C[l_u,r_u]$ are close to $u$). Let  {$1\le v\le n-1$}. \begin{itemize}
\vspace{1mm}\item If there exists  {$1\le j<n-2$ such that {$v-j\in C[l_v,v)$}  (so $v-j$ is close to $v$ from the left)} and $d(v-j,v)-d(v-j+1,v)=1=d(v-j,n)-d(v-j+1,n)$, then let  {$\Delta_l(v)$} denote such a minimum $j$; otherwise,  set  {$\Delta_l(v):=\infty$}.
\vspace{1mm}\item If there exists  {$1\le j<n-2$ such that vertex  {$v+j\in C(v,r_v]$ }(so $v+j$ is close to $v$ from the right)} and $d(v+j,v)-d(v+j-1,v)=1=d(v+j,n)-d(v+j-1,n)$, then let   {$\Delta_r(v)$} denote such a minimum $j$; otherwise,  set {$\Delta_r(v):=\infty$}.
\vspace{1mm}\item If there exists $1\le j<n-2$ such that $v-j$,  $v+j$ form  a bad pair contained in $C[l_v,r_v]$, then  let  {$\Delta(v)$} denote such a minimum $j$; otherwise, set  {$\Delta(v):=\infty$.}
 \end{itemize}

\begin{Observation} \label{ob4}
 {Let $S$ $(\subseteq V-\{n\})$ be a set of rim observers, and let $s\in S$. 
\begin{mylabel}
\vspace{0mm}\item   If 
$r_{s^-}\ge l_s-1$, then  every vertex of $C[s^-,s]$ is close to $\{s^-,s\}$. 
\vspace{0mm}    \item   If 
    $r_{s^-}\ge s-\Delta_l(s)$, then $S$ is left-continuous at $s$. 
\vspace{0mm}\item  If 
$l_{s^+}\le s+\Delta_r(s)$, then $S$ is right-continuous at $s$. 
\vspace{0mm}\item If  $r_{s^-}\ge s-\Delta(s)$, then $s^-$ covers the minimal bad pair of $s$ (if any) from the left. 
\vspace{0mm}\item   If $l_{s^+}\le s+\Delta(s)$, then $s^+$  covers the minimal bad pair of $s$ (if any) from the right. 
    \qed
\end{mylabel}}
\end{Observation}

We define two functions $F$ and $ F'$  from  $\{1,2,\ldots, n-1\}\times\{\text{left}, \text{right}\}$ to $R^+$ as follows: The functions map $(s,\cdot)$ to the minimum  weight  of a vertex set  {$S\subseteq\{1,\ldots,s\}$} satisfying
 \begin{itemize}
\vspace{-2mm}\item[(a)] $1,s\in S$;
\vspace{-2mm}\item[(b)] every vertex $i$ with  {$1\le v\le s$} is close to $S$; 
\vspace{-2mm}\item[(c)] {$S$ is left-continuous at every vertex of $S-\{1\}$, and right-continuous at every vertex of $S-\{s\}$.} 
\vspace{-2mm}\item[(d)]   {for $F(s,\text{left})$ (resp. $F'(s,\text{left})$), the minimal bad pair of any vertex in $S$ (resp.  $S-\{1\}$)  is covered by $S$, and the minimal bad pair of $s$ (if any) is covered by $s^-$ from the left,}

\vspace{-2mm}\item[(e)]   {for $F(s,\text{right})$ (resp. $F'(s,\text{right})$), the minimal bad pair of  any vertex in $  S- \{s\}$ (resp. $S-\{1,s\}$) is covered by $S$,  and the minimal bad pair of $s$ (if any) {\em will} be covered from the right later.}
 \end{itemize}
Correspondingly, for $F(s,\cdot)$ (resp. $F'(s,\cdot)$), let $S(s, \cdot)$ (resp. $S'(s,\cdot)$) denote the minimum weighted $S$ described above.

\bigskip
\hrule
\vspace{-1mm}
\begin{algorithm}\label{alg3} Finding minimum weighted DRS of general wheels
\end{algorithm}
\vspace{-1mm}\hrule

\medskip\noindent Using  dynamic program, we compute the values of $F(s,\cdot)$ and $F'(s,\cdot)$ recursively for $s=1,2,\ldots,n-1$. The initial settings are
\begin{itemize}
\item[1.] {\em Initial Step:} \\$F(1,\text{left}):=\infty$ and $F(1,\text{right})=F'(1,\text{left})=F'(1,\text{right}):=c(1)$.
    \end{itemize}
For $s>1$,  {we need select $s'$ from $\{1,2,\ldots,s-1\}$ and $S'\in\{S(s',\text{left}),S'(s',\text{left}),$ $S(s',\text{right}),S'(s',\text{right})\}$
such that $S=S'\cup\{s\}$ is the minimum weighted set satisfying (a) -- (d). In view of Observation \ref{ob4}, we define} 
\begin{align*}
&\alpha(s):=\{s'\in\{1,2,\ldots,s-1\}:r_{s'}\ge l_s-1,\; r_{s'}\ge s-\Delta_l(s),\; s'+\Delta_r(s')\ge l_s\},\\ &\beta(s):=\{s'\in \alpha(s):r_{s'}\ge s-\Delta(s)\}\subseteq\alpha(s),\\
&\gamma(s):=\{s'\in\alpha(s):l_s\le s'+\Delta(s')\}\subseteq\alpha(s).
\end{align*}
 From Observation \ref{ob4}(i)-(iii), it is easy to see that (a), (b) and (c) are satisfied as long as we take $s'$ from $\alpha(s)$. Note that $S(s,\text{left})-\{s\}$ is either $S(s',\text{left})$ or $S(s',\text{right})$ for some $s'\le s-1$, In the former case, we only need to guarantee that the minimal bad pair of $s$ (if any) is covered by $s'$. Observation \ref{ob4}(iv) says that $s'\in \beta(s)$ suffices. In the latter case,   we are done if we additionally make sure that the minimal bad pair of $s'$ (if any) is covered by $s$. This, by Observation~\ref{ob4}(v), is equivalent to requiring $s'\in \gamma(s)$. Hence we obtain
 \begin{itemize}
\item[2.] {\em Recursive Step:} \\ $F(s,\text{left})=\min\left\{\min_{s'\in\beta(s)}F(s',\text{left}),\;\min_{s'\in\beta(s)\cap\gamma(s)}F(s',\text{right})\right\}+c(s)
$
\end{itemize}
Similarly, we have other  recursive  formulas for computing $F$ and $F'$ as follows
 \begin{itemize}
\item[3.]{\em Recursive Step:} \\ $F(s,\text{right})=\min\left\{\min\limits_{s'\in\alpha(s)}F(s',\text{left}),\;
    \min\limits_{s'\in\gamma(s)}F(s',\text{right})\right\}+c(s)$
\vspace{1mm}\item[4.] {\em Recursive Step:} \\
  $F'(s,\text{left})=\min\left\{\min\limits_{s'\in\beta(s)}F'(s',\text{left}),\;\min\limits_{s'\in\beta(s)\cap\gamma(s)}F'(s',\text{right})\right\}+c(s)$
\vspace{1mm}\item[5.] {\em Recursive Step:} \\
$F'(s,\text{right})=\left\{\begin{array}{ll}c(1)+c(s),&1\!\in\!\alpha(s)\\
\min\left\{\underset{s'\in\alpha(s)}{\min}\!F'(s',\text{left}),\;
\underset{s'\in\gamma(s)}{\min}\!F'(s',\text{right})\right\}+c(s),&1\!\not\in\! \alpha(s)
\end{array}\right.$
\end{itemize}
When the recursive formulas stop at  {$s=n-1$}, we  {are ready to}  give the weight of minimum weighted DRS assuming vertex $1$ is an observer. For vertex $1$, {note that  $ 1-\Delta_l(1)\equiv  n-\Delta_l(1)\pmod {n-1}$   and $ 1-\Delta(1)\equiv  n-\Delta(1)\pmod {n-1}$.} We define the sets of   possible   consecutive observers on the left of vertex 1 as
\begin{align*}
&\alpha(1)=\{s\in\{2,3,\ldots,n-1\}:r_s\ge l_1-1,\; r_s\ge n-\Delta_l(1),\; s+\Delta_r(s)\ge l_1\},\\
&\beta(1)=\{s\in \alpha(1):r_s\ge n-\Delta(1)\},\\
&\gamma(1)=\{s\in\alpha(1):l_1\le s+\Delta(s)\}.
\end{align*}
 {In particular, picking $s=1^-$ from $\alpha(1)$ guarantees that $S$ is left-continuous at~$1$.} {Furthermore, combining (a)--(e) and  {Observation \ref{ob4}}, we deduce from Theorem \ref{nwhl} that the minimum weighted DRS of $G$ must be one of $S(s,\text{left})$ with $s\in \alpha(1)$,  $S'(s,\text{left})$ with $s\in \beta(1)$,  $S(s,\text{right})$ with $s\in \gamma(1)$,  and $S'(s,\text{right})$ with $s\in \beta(1)\cap\gamma(1)$. These sets satisfy all the conditions in Theorem \ref{nwhl}. Hence the weight of the optimal DRS is given by
 { \begin{itemize}
\item[6.] {\em Final Step:} \\
$\min\left\{\min\limits_{s\in\alpha(1)}\!\!F(s,\text{left}), \min\limits_{s\in\beta(1)}\!\!F'(s,\text{left}), \min\limits_{s\in\gamma(1)}\!\!F(s,\text{right}), \min\limits_{s\in\beta(1)\cap\gamma(1)}\!\!F'(s,\text{right})\right\}.$

By backtracking we can get the optimal set.\end{itemize}}
\vspace{-1mm}\hrule
 }

 \medskip\medskip {{To summarize, Steps 1-6 present our dynamic programming algorithm for finding   minimum weighted DRS  in wheels. Its correctness follows from the description and corresponding analysis.}

\begin{theorem}
The {MWDRS} problem on  general wheels can be solved in  $O(n^3)$ time.\qed
\end{theorem}

\end{document}